\documentclass[a4paper]{article}

\usepackage[utf8]{inputenc}
\usepackage[twocolumn,a4paper,left=1.2cm,right=1.2cm,top=2.5cm,bottom=2.5cm,columnsep=.7cm]{geometry}
\usepackage{booktabs} 
\usepackage{multirow}
\usepackage{amsmath,amssymb}
\usepackage{amsthm}
\usepackage{graphicx}
\usepackage[ruled,linesnumbered]{algorithm2e}
\usepackage{hyperref}
\hypersetup{
  colorlinks,
  citecolor=black,
  filecolor=black,
  linkcolor=black,
  urlcolor=black
}

\usepackage{balance}
\usepackage{abstract}

\newtheorem{definition}{Definition}

\newtheorem{theorem}{Theorem}
\newtheorem{lemma}{Lemma}

\DeclareMathOperator*{\argmax}{arg\,max}

\newcommand{\algo}[1]{\textsc{#1}}

\newcommand{\BO}[1]{O\left(#1\right)}

\newcommand{\BOM}[1]{\Omega\left(#1\right)}

\renewcommand{\epsilon}{\varepsilon}

\newcommand{\Let}[2]{#1 $\leftarrow$ #2}
\SetKw{To}{to}
\SetKw{Error}{Error}
\SetKwComment{Comment}{$\triangleright$\ }{}
\SetKwProg{Procedure}{Procedure}{ }{end}

\newif\iflong
\longtrue

\title{Clustering Uncertain Graphs}
\date{}

\usepackage{titling}

\usepackage{authblk}
\newcommand{\email}[1]{\texttt{#1}} 

\author[1]{Matteo Ceccarello}
\author[1]{Carlo Fantozzi}
\author[1]{Andrea Pietracaprina}
\author[1]{\authorcr Geppino Pucci}
\author[1]{Fabio Vandin}
\affil[1]{
  Department of Information Engineering\\
  University of Padova, Italy\\
  \email{\{ceccarel,fantozzi,capri,geppo,vandinfa\}@dei.unipd.it}
}

\newcommand{\abstracttext}{
An uncertain graph $\mathcal{G} = (V, E, p : E \rightarrow (0,1])$ can
be viewed as a probability space whose outcomes (referred to as
\emph{possible worlds}) are subgraphs of $\mathcal{G}$ where any edge
$e\in E$ occurs with probability $p(e)$, independently of the other
edges. These graphs naturally arise in many application domains where
data management systems are required to cope with uncertainty in
interrelated data, such as computational biology, social network
analysis, network reliability, and privacy enforcement, among the
others. For this reason, it is important to devise fundamental
querying and mining primitives for uncertain graphs. This paper
contributes to this endeavor with the development of novel strategies
for clustering uncertain graphs. Specifically, given an uncertain
graph $\mathcal{G}$ and an integer $k$, we aim at partitioning its
nodes into $k$ clusters, each featuring a distinguished center node,
so to maximize the minimum/average connection probability of any node
to its cluster's center, in a random possible world.  We assess the
NP-hardness of maximizing the minimum connection probability, even in
the presence of an oracle for the connection probabilities, and
develop efficient approximation algorithms for both problems and 
some useful variants. Unlike previous works in the literature, our
algorithms feature provable approximation guarantees and are capable
to keep the granularity of the returned clustering under control. Our
theoretical findings are complemented with several
experiments that compare our algorithms against some 
relevant competitors, with respect to both running-time and
quality of the returned clusterings.  
}

\begin{document}

\twocolumn[
\maketitle
\begin{abstract}
  \abstracttext{}
\end{abstract}
\vspace{2em}
] 

\section{Introduction}
\label{sec:intro}

In the big data era, data management systems are often required to
cope with uncertainty \cite{Aggarwal10}.  Also, many application
domains increasingly produce interrelated data, where uncertainty may
concern the intensity or the confidence of the relations between
individual data objects. In these cases, graphs provide a natural
representation for the data, with the uncertainty modeled by
associating an existence probability to each edge. For example, in
Protein-Protein Interaction (PPI) networks, an edge between two
proteins corresponds to an interaction that is observed through a
noisy experiment characterized by some level of uncertainty, which can
thus be conveniently cast as the probability of existence of that edge
\cite{AsthanaKGR04}. Also, in social networks, the probability of
existence of an edge between two individuals may be used to model the
likelihood of an interaction between the two individuals, or the
influence of one of the two over the other \cite{Kempe2003,AdarR07}.
Other applications of uncertainty in graphs arise in the realm of
mobile ad-hoc networks \cite{BiswasM05,GhoshNYQ07}, knowledge
bases~\cite{Bollacker2008}, and graph
obfuscation for privacy enforcement \cite{BoldiBGT12}. This variety of
application scenarios calls for the development of fundamental
querying and mining primitives for uncertain graphs which, as argued later, can become computationally challenging
even for graphs of moderate size.

Following the mainstream literature \cite{PotamiasBGK10}, an
\emph{uncertain graph} is defined over a set of nodes $V$, a set of
edges $E$ between nodes of $V$, and a probability function $p : E
\rightarrow (0,1]$. We denote such a graph by $\mathcal{G}=(V, E, p)$.
  $\mathcal{G}$ can be viewed as a probability space whose outcomes
  (referred to as \emph{possible worlds}, in accordance with the
  terminology adopted for probabilistic databases
  \cite{BenjellounSHW06,DalviS07}) are graphs $G=(V,E')$ where any
  edge $e \in E$ is included in $E'$ with probability $p(e)$,
  independently of the other edges. The main objective of this work is
  to introduce novel strategies for clustering uncertain graphs,
  aiming at partitioning the node set $V$ so to maximize two
connectivity-related objective functions, which can be seen as
reinterpretations of the objective functions of the classical
$k$-center and $k$-median problems \cite{Schaeffer07} in the framework
of uncertain graphs.

In the next subsection we provide a brief account on the literature on
uncertain graphs most relevant to our work.

\subsection{Related work}
Early work on network reliability has dealt implicitly with the
concept of uncertain graph.  In general, given an uncertain graph, if
we interpret edge probabilities as the complement of failure
probabilities, a typical objective of network reliability analysis is
to determine the probability that a given set of nodes is connected
under random failures. This probability can be estimated through a Monte
Carlo approach, which however becomes prohibitively cumbersome for very
low reliability values. In fact, even the simplest problem of computing
the exact probability that two distinguished nodes $s$ and $t$ are
connected is known to be $\#P$-complete \cite{Ball86,Valiant79}.
In the last three decades, several works have
tried to come up with better heuristics for various reliability
problems on uncertain graphs (see \cite{JinLDW11} and references
therein). Some works have studied various formulations of the problem
of determining the most reliable source in a network subject to edge
failures, which are special cases of the clustering problems studied
in this paper (see \cite{Ding11} and references therein).

The definition of uncertain graph adopted in this paper has been
introduced in \cite{PotamiasBGK10}, where the authors investigate
various probabilistic notions of distance between nodes, and develop
efficient algorithms for determining the $k$ nearest neighbors of a
given source under their different distance measures. It has to be
remarked that the proposed measures do not satisfy the triangle 
inequality, thus ruling out the applicability of traditional metric
clustering approaches. In the last few years, there has been a
multitude of works studying several analytic and mining problems on
uncertain graphs. A detailed account of the state of the art on the
subject can be found in \cite{ParchasGPB15}, where the authors also
investigate the problem of extracting a representative possible world
providing a good summary of an uncertain graph for the purposes of
query processing.

A number of recent works have studied different ways of clustering 
uncertain graphs, which is the focus of this paper. In
\cite{KolliosPT13} the authors consider, as a clustering problem, the
identification of a deterministic \emph{cluster graph}, which
corresponds to a clique-cover of the nodes of the uncertain graph,
aiming at minimizing the expected \emph{edit distance} between the
clique-cover and a random possible world of the uncertain graph, where
the edit distance is measured in terms of edge additions and
deletions. A 5-approximation algorithm for this problem is provided in~
\cite{KolliosPT13}. The main drawback of this approach is that the formulation
of the clustering problem does not allow to control the number of
clusters. Moreover, the approximate solution returned by the proposed
algorithm relies on a shallow star-decomposition of the topology of
the uncertain graph, which always yields a large number of clusters
(at least $|V|/(\Delta+1)$, where $\Delta$ is the maximum degree of a
node in $V$).  Thus, the returned clustering may not exploit
more global information about the connectivity properties of the
underlying topology.

The same clustering problem considered in \cite{KolliosPT13} has been
also studied by Gu et al. in \cite{GuGCY14} for a more general class of
uncertain graphs, where the assumption of edge independence is lifted
and the existence of an edge $(u,v)$ is correlated to the existence of
its adjacent edges (i.e., edges incident on either $u$ or $v$).  The
authors propose two algorithms for this problem, one that, as in
\cite{KolliosPT13}, does not fix a bound on the number of returned
clusters, and another that fixes such a bound. Neither algorithm is
shown to provide worst-case guarantees on the approximation ratio.
 
In \cite{LiuJAS12} a clustering problem is defined with the objective
of minimizing the expected entropy of the returned clustering, defined
with respect to the adherence of the clustering to the connected
components of a random possible world. With this objective in mind,
the authors develop a clustering algorithm which combines a standard
$k$-means strategy with the Monte Carlo sampling approach for
reliability estimation. No theoretical guarantee is offered on the
quality of the returned clustering with respect to the defined
objective function, and the complexity of the approach, which does not
appear to scale well with the size of the graph, also depends on a
convergence parameter which cannot be estimated analytically. In
summary, while the pursued approach to clustering has merit, there is
no rigorous analysis of the tradeoffs that can be exercised between
the quality of the returned clustering and the running time of the
algorithm.

In~\cite{Dongen08} the \emph{Markov Cluster Algorithm} (\algo{mcl}) is
proposed for clustering weighted graphs.  In \algo{mcl}, an edge
weight is considered as a \emph{similarity score} between the
endpoints.  The algorithm does not specifically target uncertain
graphs, but it can be run on these graphs by considering the edge
probabilities as weights.  In fact, some of the aforementioned works
on the clustering of uncertain graphs have used \algo{mcl} for comparison
purposes.  The algorithm focuses on finding \mbox{so-called}
\emph{natural clusters}, that is, sets of nodes characterized by the
presence of many edges and paths between their members.  The basic
idea of the algorithm is to perform random walks on the graph and
to partition the nodes into clusters according to the probability of a
random walk to stay within a given cluster. Edge weights (i.e., the
similarity scores) are used by the algorithm to define the probability
that a given random walk traverses a given edge.  The algorithm's
behaviour is controlled with a parameter, called \emph{inflation},
which indirectly controls the granularity of the clustering. However,
there is no fixed analytic relation between the inflation parameter
and the number of returned clusters, since the impact of the inflation
parameter is heavily dependent on the graph's topology and on the edge
weights.  The author of the algorithm maintains an optimized and very
efficient implementation of \algo{mcl}, against which we will compare
our algorithms in Section~\ref{sec:experiments}.

Finally, it is worth mentioning that the problem of influence
maximization in a social network under the Independent Cascade model,
introduced in \cite{Kempe2003}, can be reformulated as the search of
$k$ nodes that maximize the expected number of nodes reachable from
them on an uncertain graph associated with the social network, where
the probability on an edge $(u,v)$ represents the likelihood of $u$
influencing $v$. A constant approximation algorithm for this problem,
based on a computationally heavy Monte Carlo sampling, has been
developed in \cite{Kempe2003}, and a number of subsequent works have
targeted faster approximation algorithms (see \cite{BorgsBCL14,Tang2015} and
references therein). It is not clear whether the solution of this
problem can be employed to partition the nodes into $k$ clusters that
provide good approximations for our two objective functions.

\subsection{Our Contribution}
\label{sec:contrib}
In this paper we develop novel strategies for clustering uncertain
graphs.  As observed in \cite{LiuJAS12}, a good clustering should aim
at high probabilities of connectivity within clusters, which is
however a hard goal to pursue, both because of the inherent difficulty
of clustering per se, and because of the aforementioned
\#P-completeness of reliability estimation in the specific uncertain
graph scenario.  Also, it has been observed
\cite{PotamiasBGK10,LiuJAS12} that the straightforward reduction to
shortest-path based clustering where edge probabilities become weights
may not yield significant outcomes because it disregards the possible
world semantics.

Motivated by the above scenario, we will adopt the \emph{connection
  probability} between two nodes (a.k.a. two-terminal reliability),
that is, the probability that the two nodes belong to the same
connected component in a random possible world, as the distance
measure upon which we will base our clustering. As a first technical
contribution, which may be of independent interest for the broader
area of network reliability, we show that this measure satisfies a
form of triangle inequality, unlike other distance measures used in
previous works.  This property allows us to cast the problem of
clustering uncertain graphs into the same framework as traditional
clustering approaches on metric spaces, while still enabling an
effective integration with the possible world semantics that must be
taken into account in the estimation of the connection probability.

Specifically, we study two clustering problems, together with some
variations. Given in input an $n$-node uncertain graph $\mathcal{G}$ and an
integer $k$, we seek to partition the nodes of $\mathcal{G}$ into $k$
clusters, where each cluster contains a distinguished node, called
\emph{center}. We will devise approximation algorithms for each of the
following two optimization problems: (a) maximize the Minimum
Connection Probability of a node to its cluster center (MCP problem);
and (b) maximize the Average Connection Probability of a node to its
cluster center (ACP problem).

We first prove that the MCP problem is NP-hard even in the presence of
an oracle for the connection probability, and make the plausible
conjecture that the ACP problem is NP-hard as well.
Our approximation algorithms for the MCP and ACP problems
are both based on a simple deterministic
strategy that computes a \emph{partial} $k$-clustering aiming at
covering a  maximal subset of nodes, given a threshold on the
minimum connection probability of a node to its cluster's center.
By incorporating this strategy within suitable guessing schedules,
we are able to obtain $k$-clusterings with the following guarantees:
\begin{itemize}
\item
  \sloppy
for the MCP problem, minimum connection probability 
$\BOM{p^2_{\rm opt-min}(k)}$, where $p_{\rm opt-min}(k)$ is the maximum
minimum connection probability of any $k$-clustering.
\item
for the ACP problem, average connection probability 
$\BOM{(p_{\rm opt-avg}(k)/\log n)^3}$, where $p_{\rm opt-avg}(k)$ is the maximum
average connection probability of any $k$-clustering.
\end{itemize}
We also discuss variants of our algorithms that allow 
to impose a limit on the length
of the paths that contribute to the connection probability between two
nodes. Computing provably good clusterings under limited path length
may have interesting applications in scenarios such as the analysis of
PPI networks, where topological distance between two nodes diminishes
their similarity, regardless of their connection probability.

We first present our clustering algorithms assuming the availability
of an oracle for the connection probabilities between pairs of nodes.
Then, we show how to integrate a progressive sampling scheme for the
Monte Carlo estimation of the required probabilities, which essentially
preserves the approximation quality.  Recall that, due to
the \#P-completeness of two-terminal reliability, the Monte Carlo
estimation of connection probabilities is computationally intensive
for very small values of these probabilities. A key feature of our
approximation algorithms is that they only require the estimation of
probabilities not much smaller than the optimal value of the objective
functions.  In the case of the MCP problem, this is achieved by a
simple adaptation of an existing approximation algorithm for the
$k$-center problem \cite{HochbaumS85}, while for the ACP problem our
strategy to avoid estimating small connection probabilities is novel.

To the best of our knowledge, ours are the first clustering algorithms
for uncertain graphs that are fully parametric in the number of
desired clusters while offering provable guarantees with respect to the optimal
solution, together with efficient implementations. While the theoretical
bounds on the approximation are somewhat loose, especially for small values
of the optimum, we report the results of a number of experiments showing
that in practice the quality of the clusterings returned by our
algorithms is very high.
In particular, we perform an experimental comparison of our algorithms with \algo{mcl} which, as discussed above, is widely used in the context of uncertain graphs.
We also compare with a naive adaptation of a classic $k$-center algorithm to verify that such adaptations lead to poor results, prompting for the development of specialized algorithms.
We run our experiments on uncertain graphs derived from PPI networks and on a large collaboration graph derived from DBLP, finding that our algorithms identify good clusterings
with respect to both our metrics, while a clustering strategy that is not
specifically designed for uncertain graphs, as the one employed by
\algo{MCL}, may in some cases provide a very poor clustering.
Moreover, on PPI networks, we evaluate the performance of our algorithms
in predicting protein complexes, finding that we can obtain results
comparable with state-of-the-art solutions.

The rest of the paper is organized as follows. In
Section~\ref{sec:prelims} we define basic concepts regarding uncertain
graphs and formalize the MCP and ACP problems. Also, we prove
a form of triangle inequality for the connection probability measure
and discuss the NP-hardness of the two problems.  In
Section~\ref{sec:algs} we describe and analyze our clustering
algorithms. In Section~\ref{sec:implementation} we show how to
integrate the Monte Carlo probability estimation within the algorithms
while maintaining comparable approximation guarantees.
Section~\ref{sec:experiments} reports the results of the experiments.
Finally,  Section~\ref{sec:conclusions} offers some concluding remarks and
discusses possible avenues of future research.


\section{Preliminaries}
\label{sec:prelims}
Let $\mathcal{G}=(V, E, p)$ be an uncertain graph, as defined in the
introduction. In accordance with the established notation used in
previous work, we write $G\sqsubseteq \mathcal{G}$ to denote that $G$
is a possible world of $\mathcal{G}$. Given two nodes $u, v \in V$,
the probability that they are connected (an event denoted as $u\sim
v$) in a random possible world can be defined as
\[
  \Pr(u\sim v) =
  \sum_{G\sqsubseteq \mathcal{G}} \Pr(G) \textbf{I}_{G}(u, v),
\]
where 
\[
  \textbf{I}_G(u, v) =
  \left\{ 
    \begin{aligned}
      1 &\quad \text{ if } u \sim v \text{ in } G \\
      0 &\quad \text{ otherwise}
    \end{aligned}
  \right.
\]
We refer to $\Pr(u \sim v)$ as the \emph{connection probability}
between $u$ and $v$ in $\mathcal{G}$. The uncertain graphs we
consider in this paper, hence their possible worlds, are undirected
and, except for the edge probabilities, no weights are attached to
their nodes/edges.

Given an integer $k$, with $1 \leq k < n$, a
\emph{$k$-clustering} of $\mathcal{G}$ is a partition of $V$ into
$k$ \emph{clusters} $C_1, \dots, C_k$ and a set of
\emph{centers} $c_1, \dots, c_k$ with $c_i \in C_i$, for $1 \leq i
\leq k$. We aim at clusterings where each node is
well connected to its cluster center in a random possible world. To
this purpose, for
a $k$-clustering $\mathcal{C}=(C_1,\ldots,C_k; c_1,\ldots,c_k)$
of $\mathcal{G}$ 
we define the following two objective functions
\begin{align} 
\mbox{min-prob}(\mathcal{C}) & = 
\min_{1\le i \le k} \min_{v \in C_i} \Pr(c_i \sim v), \label{pmin} \\
\mbox{avg-prob}(\mathcal{C}) & = 
(1/n) \sum_{1\le i \le k} \sum_{v \in C_i} \Pr(c_i \sim v). \label{pmedian}
\end{align}
\begin{definition}
Given an uncertain graph $\mathcal{G}$ with $n$ nodes and
an integer $k$, with $1 \leq k < n$,
the \emph{Minimum Connection Probability (MCP)} 
\emph{({\em resp.,} Average Connection Probability (ACP))} problem
requires to determine a $k$-clustering $\mathcal{C}$
of $\mathcal{G}$ with maximum 
\emph{$\mbox{min-prob}(\mathcal{C})$
({\em resp.,} $\mbox{\rm avg-prob}(\mathcal{C})$)}.
\end{definition}

By defining the distance between two nodes $u,v \in V$ as $d(u,v) =
\ln (1/\Pr(u \sim v))$, with the understanding that $d(u,v)=\infty$ if
$\Pr(u \sim v)=0$, it is easy to see that the $k$-clustering that
maximizes the objective function given in Equation (\ref{pmin})
(resp., (\ref{pmedian})) also minimizes the maximum distance (resp.,
the average distance) of a node from its cluster center. Therefore,
the MCP and ACP problems can be reformulated as instances of the
well-known NP-hard $k$-center and $k$-median problems
\cite{Vazirani01}, which makes the former the direct counterparts
of the latter in the realm of uncertain graphs. However, objective functions that
exercise alternative combinations of minimization and averaging of
connection probabilities are in fact possible, and we leave their
exploration as an interesting open problem.


While the $k$-center/median problems remain NP-hard even when the
distance function defines a metric space, thus, in particular,
satisfying the triangle inequality (i.e, $d(u,z) \leq d(u,v)+d(v,z)$),
this assumption is crucially exploited by most approximation
strategies known in the literature. Therefore, in order to port these
strategies to the context of uncertain graphs, we need to show that
the distances derived from the connection probabilities, as explained
above, satisfy the triangle inequality. This is equivalent to showing
that for any three nodes $u,v,z$, $\Pr(u \sim z) \geq \Pr(u \sim v)
\cdot \Pr(v \sim z)$ , which we prove below.

Fix an arbitrary edge $e \in E$ and let $A(e)$ be the event: ``edge
$e$ is present''. We need the following technical lemma.
\begin{lemma} \label{conditional1}
For any pair $x,y \in V$, we have
\[\Pr(x \sim y | A(e)) \geq \Pr(x \sim y | \neg A(e))\]
\end{lemma}
\begin{proof}
Let $\mathcal{G}^{x,y}_e$ (resp., $\mathcal{G}^{x,y}_{\neg e}$)
be the set of possible worlds where
$x \sim y$, and edge $e$ is present (resp., not present). We have that
\begin{eqnarray*}
\Pr(x \sim y | A(e)) 
& = &\sum_{G \sqsubseteq \mathcal{G}^{x,y}_e} \Pr(G)/p(e), \\
\Pr(x \sim y | \neg A(e)) 
& = &\sum_{G \sqsubseteq \mathcal{G}^{x,y}_{\neg e}} \Pr(G)/(1-p(e)). 
\end{eqnarray*}
The lemma follows by observing that for any graph $G$ in
$\mathcal{G}^{x,y}_{\neg e}$ the same graph with the addition of $e$
belongs to $\mathcal{G}^{x,y}_e$, and the corresponding terms in the
two summations are equal.
\end{proof}
\begin{theorem} \label{triangle}
For any uncertain graph
$\mathcal{G}=(V, E, p)$ and any triplet $u,v,z \in V$, we have:
\[
\Pr(u \sim z) \geq \Pr(u \sim v) \cdot \Pr(v \sim z).
\]
\end{theorem}
\begin{proof}
The proof proceeds by induction on the number $k$ of \emph{uncertain
edges}, that is, edges $e \in E$ with $p(e) > 0$ and $p(e) < 1$.
Fix three nodes $u,v,z \in V$. The base case $k=0$ is trivial:
in this case, the uncertain graph is deterministic and for each
pair of nodes $x,y \in V$, $\Pr(x \sim y)$ is either 1 or 0, which implies
that when $\Pr(u \sim v) \cdot \Pr(v \sim z) = 1$, then $\Pr(u \sim z)
=1$ as well. Suppose that the property holds for uncertain
graphs with at most $k$ uncertain edges, with $k \geq 0$, and consider
an uncertain graph $\mathcal{G}=(V, E, p)$ with $k+1$ uncertain
edges. Fix an arbitrary edge $e \in E$ and let $A(e)$
denote the event that edge $e$ is present. For any two arbitrary nodes
$x,y \in V$, we can write
\begin{equation*}
  \begin{aligned}
    \Pr(x \sim   y) 
    &= \Pr(x \sim   y | A(e)) \cdot p(e)\\
    &\qquad+\Pr(x \sim y | \neg A(e)) \cdot (1-p(e))\\
    &= \left(\Pr(x \sim y | A(e)) - \Pr(x \sim y | \neg A(e))\right) \cdot p(e)\\
    &\qquad+\Pr(x \sim y | \neg A(e)).
  \end{aligned}
\end{equation*}
By Lemma~\ref{conditional1}, the term multiplying $p(e)$ in the
above expression is nonnegative. As a consequence, we have that
\[
\Pr(u \sim v) \cdot 
\Pr(v \sim z) - 
\Pr(u \sim z) =
A \cdot (p(e))^2 + B \cdot p(e) + C,
\]
for some constants $A,B,C$ independent of $p(e)$, with $A \geq
0$. Therefore, the maximum value of $\Pr(u \sim v) \cdot 
\Pr(v \sim z) - \Pr(u \sim z)$, as a
function of $p(e)$, is attained for $p(e)=0$ or $p(e) = 1$. 
Since in either case the number of uncertain edges is decremented
by one, by the inductive hypothesis, the difference must yield
a nonpositive value, hence the theorem follows. 
\end{proof}

A fundamental primitive required for obtaining the desired clustering
is the estimation of $\Pr(u\sim v)$ for any two nodes $u,v \in V$.
While the exact computation of $\Pr(u\sim v)$ is
$\#P$-complete \cite{Ball86}, for reasonably large values of this 
probability a very accurate estimate can be obtained through
Monte Carlo sampling. More precisely, for $r >0$
let $G_1, \dots, G_r$ be $r$ sample possible worlds 
drawn independently at random from
$\mathcal{G}$. For any pair of nodes $u$ and $v$
we can define the following estimator
\begin{equation}\label{eq:estimator}
\tilde{p}(u, v) =\frac{1}{r} \sum_{i=1}^r \textbf{I}_{G_i}(u, v)
\end{equation}
It is easy to see that $\tilde{p}(u, v)$ is an unbiased estimator of
$\Pr(u\sim v)$. Moreover, by taking
\begin{equation}
  \label{eq:num-samples}
  r\ge \frac{3\ln \frac{2}{\delta}}{\epsilon^2 \Pr(u\sim v)}
\end{equation}
samples, we have that $\tilde{p}(u, v)$ is an
$(\epsilon,\delta)$-approximation of $\Pr(u\sim v)$, that is,
\begin{equation}
\label{eq:eps-delta-approx}
\Pr \left( \frac{|\tilde{p}(u, v)- \Pr(u\sim v)|}{\Pr(u\sim v)} \leq \epsilon \right) 
\geq 1-\delta
\end{equation}
(e.g., see \cite[Theorem 10.1]{MitzenmacherU05}). This approach is very
effective when $\Pr(u\sim v)$ is not very small. However, when $\Pr(u\sim v)$ 
is small (i.e., it approaches 0), the number of samples, hence the work,
required to attain an accurate estimation becomes prohibitively large.

Even if the probabilities $\Pr(u\sim v)$ were provided by an oracle
(i.e., they could be computed efficiently), the MCP problem remains
computationally difficult. Indeed, consider the following decision
problem: given an uncertain graph $\mathcal{G}=(V, E, p)$, an oracle
for estimating pariwise connection probabilities, an integer $k \ge
1$, and $\hat{p}$ with $0 \le \hat{p} \le 1$, is there a
$k$-clustering $\mathcal{C}$ such that $\mbox{min-prob}(\mathcal{C})
\ge \hat{p}$? We have:
\begin{theorem} \label{thm:NPhard}
The decision problem above is NP-hard.
\end{theorem}
\iflong
\begin{proof}
The proof is by reduction from set cover. Consider an instance of set cover, with $U=\{u_1,\dots,u_m\}$ denoting 
the universe of elements and $\mathcal{S} =\{ S_1,\dots,S_n \}$
 being a family of subsets of $U$, each
\emph{covering} its own elements. Given an integer $k$, the set cover
problem asks whether there are $k$ sets that cover all the elements of
$U$, that is, $S_{i_1}, S_{i_2},\dots, S_{i_k}$, with $S_{i_j} \in
\mathcal{S}$ for $1 \le j \le k$, such that $\cup_{j=1}^k S_{i_j}
= U$. Note that for each element $u \in U$ there must be a subset $S\in \mathcal{S}$ 
such that $u \in S$, or otherwise there cannot be a set cover for
$U$. Since such condition can be checked in polynomial time, we may
assume that this is always the case.

Given an instance of set cover, we build an instance of our decision
problem as follows. Let $N = m+n$. We consider the uncertain graph
$\mathcal{G}=(V,E,p)$ where: $V = U \cup \mathcal{S}$; $E = \{(u,S): S
\in \mathcal{S}, u \in S\} \cup \{(S,S'): S,S' \in \mathcal{S}\}$;
$p(e) = \frac{1}{N!}$. Note that this procedure takes time polynomial
in $m \cdot n$.

We now prove that $\mathcal{G}=(V,E,p)$ admits a $k$-clustering such
that $\min_{1\le i \le k} \min_{v \in C_i} \Pr(c_i \sim v) \ge
\frac{1}{N!}$ if and only if there is a set cover of cardinality
$k$. We first prove that if there is a set cover $S_{i_1}, S_{i_2}, \dots, S_{i_k}$ of cardinality $k$,
then there is a clustering with minimum probability $\ge
\frac{1}{N!}$. Define $S_{i_1}, S_{i_2}, \dots, S_{i_k}$ to be centers
of the $k$ clusters, and assign each vertex $u$ in $U$ to a cluster of
center $S_{i_j}$ with $u \in S_{i_j}$ (note that one such center must
exist since $S_{i_1}, S_{i_2}, \dots, S_{i_k}$ is a set cover); also, each
vertex $S \in \mathcal{S} \setminus \{ S_{i_1}, S_{i_2}, \dots,
S_{i_k} \}$ is assigned to an arbitrary cluster. Note that by
construction, for each vertex $u \in V \setminus \{ S_{i_1}, S_{i_2},
\dots, S_{i_k} \} $ and its corresponding center $S_{i_j}$ in the
clustering, $(u,S_{i_j}) \in E$; therefore $ \Pr(u \sim S_{i_j}) \ge
\frac{1}{N!}$.  Thus,
 the minimum connection probability for the
clustering is $\ge \frac{1}{N!}$.

We now prove that if there is a $k$-clustering with minimum connection
probability $\ge \frac{1}{N!}$ then there is a set cover of size $k$ . We
first prove that in every $k$-clustering with minimum connection
probability $\ge \frac{1}{N!}$ a vertex $v$ of a cluster is directly
connected to the center $c$ of its cluster, that is $(v,c) \in E$.
Let us assume that this is not the case, and consider the node $u$ and
the center $c$ of its cluster such that $(u,c) \not\in E$. Note that
the event ``$u$ is connected to $c$" is equivalent to ``$\cup_{p \in
SP}$~$p$ exists'', where $SP$ is the set of all simple paths among
$u$ and $c$ in $\mathcal{G}$. Thus, by union bound $\Pr(u \sim c) \le
\sum_{p \in SP} \Pr(p)$, with $\Pr(p)$ being the probability that path
$p$ exists in a realization of $\mathcal{G}$. Note that there are
$<N$ paths of length $2$ in $SP$ each with $\Pr(p) =
\left(\frac{1}{N!}\right)^2$, and $< N \cdot N!$ paths of length $>2$
in $SP$, each with probability $\Pr(p) \le
\left(\frac{1}{N!}\right)^3$. Therefore, if $u$ and $c$ are not directly
connected $\Pr(u \sim c) < \frac{1}{N!}$, that is, a contradiction. If
all the centers of the clustering are in the set $\mathcal{S}$, then
the centers cover the universe $U$ of elements: for every vertex
$u \in U$, the center of the cluster $u$ belongs to contains $u$.
Otherwise, we can transform the clustering in a clustering with
all centers in $\mathcal{S}$ and the same minimum connection
probability: consider a center $c \in U$, and let $C$ be the set of
vertices other than $c$ in the cluster with center $c$.
If $C = \emptyset$, reassign the center of cluster $C\cup \{c\}$
to an arbitrarily chosen element of $\mathcal{S}$ that
covers $c$. (If all elements of $\mathcal{S}$ covering $c$ are
already centers, assign $c$ to an arbitrary cluster among the ones
with such centers.) Otherwise, note that $C$ contains only vertices
of $\mathcal{S}$, since there is no edge among vertices in $U$. Now
reassign the center of cluster $C \cup \{c\}$ to an arbitrary vertex
of $C$; note that since all vertices in $C$ are connected by an edge,
the minimum connection probability for the new clustering is $\ge
1 / N!$.
\end{proof}
We remark that the NP-hardness of our clustering problem on uncertain
graphs does not follow immediately from the transformation of
connection probabilities into distances mentioned earlier, which
yields instances of the standard NP-hard $k$-center clustering
problem, since such a transformation only shows that our problem is a
restriction of the latter, where distances have the extra constraint
to be derived from connection probabilities in the underlying
uncertain graph. 
\fi

We conjecture that a similar hardness result can be proved for the
decision version of the ACP problem. Evidence in this
direction is provided by the fact that by modifying both the MCP and ACP
problems to feature a parametric upper limit on the lengths of the
paths contributing to the connection probabilities, a variant which we
study in Section~\ref{sec-depthlimit}, NP-hardness results for both the
modified problems follow straightforwardly (i.e., when paths of length at most $1$ are considered) from the hardness of $k$-center and $k$-median clustering.


\section{Clustering algorithms}
\label{sec:algs}
A natural approach to finding good solutions for the MCP and ACP
problems would be to resort to the well-known approximation strategies
for the distance-based counterparts of these problems~\cite{Schaeffer07}.
However, straightforward implementations of these strategies
may require the computation of exact connection probabilities (to be transformed into
distances) between arbitrary pairs of nodes, which can in principle be
rather small. As an example, the popular $k$-center clustering strategy
devised in \cite{Gonzalez85} relies on the iterated selection of the next center
as the {\em farthest} point from the set of currently selected ones,
which corresponds to the determination of the node featuring the smallest connection probability
to any node in the set,  when adapted to the uncertain graph scenario.
As we pointed out in the previous section, the exact computation of
connection probabilities, especially if very small, is a
computationally hard task. Therefore, for uncertain graphs we must
resort to clustering strategies that are robust to approximations and try to avoid altogether
the estimation of very small connection probabilities.

To address the above challenge, in Subsection~\ref{sec:outliers} we introduce a useful primitive
that, given a threshold $q$ on the connection probability, returns a
partial $k$-clustering of an uncertain graph $\mathcal{G}$ where the
clusters cover a maximal subset of nodes, each connected to its
cluster center with probability at least $q$, while all other nodes,
deemed \emph{outliers}, remain uncovered.  In
Subsections~\ref{sec:k-center-algo} and \ref{sec:k-median-algo} we use
such a primtive to derive  approximation algorithms for the MCP
and ACP problems, respectively, which feature provable guarantees on
the quality of the approximation and lower bounds on the value of the
connection probabilities that must be ever estimated.
We also show how the approximation guarantees of the proposed
algorithms change when connection probabilities are defined only with
respect to paths of limited length.

All algorithms presented in this section take as input an uncertain
graph $\mathcal{G}=(V, E, p)$ with $n$ nodes, and assume the existence
of an oracle that given two nodes $u,v \in V$ returns $\Pr(u \sim
v)$. In Section~\ref{sec:implementation} we will discuss how to
integrate the Monte Carlo estimation of the connection probabilities
within our algorithms.

\subsection{Partial clustering}
\label{sec:outliers}
A \emph{partial $k$-clustering} $\mathcal{C}=(C_1,\ldots,C_k;
c_1,\ldots,c_k)$ of $\mathcal{G}$ is
a partition of a subset of $V$
into $k$ clusters $C_1,\ldots,C_k$ (i.e., $\cup_{i=1,k} C_i \subseteq V$),
where each cluster $C_i$ is centered at $c_i \in C_i$, for $1 \leq i \leq
k$. We can still define $\mbox{min-prob}(\mathcal{C})$ as in
Equation~\ref{pmin} with the understanding that the uncovered nodes
(i.e., $V-\cup_{i=1,k} C_i$) are not accounted for in $\mbox{min-prob}(\mathcal{C})$.
In what follows, the term \emph{full} $k$-clustering
or, simply, $k$-clustering will refer only to a $k$-clustering covering all nodes.

\sloppy
The following algorithm, called \algo{min-partial}
(Algorithm~\ref{alg:outliers} in the box), computes a partial
  $k$-clustering $\mathcal{C}$ of $\mathcal{G}$ with
$\mbox{min-prob}(\mathcal{C}) \geq q$  covering a maximal set of
nodes, in the sense that all nodes uncovered by the clusters have
probability less than $q$ of being connected to any of the cluster
centers. The algorithm is based on a generalization of the strategy
introduced in \cite{CharikarKMN01} and uses two design parameters
$\alpha, \bar{q}$, where $\alpha \geq 1$ is an integer
and $\bar{q} \in [q,1]$, which are employed to exercise
suitable tradeoffs between performance and approximation quality. In  each of the
$k$ iterations, \algo{min-partial} picks a suitable new center as follows.
Let $V'$ denote the nodes connected with probability less than $q$
to the set of centers $S$ selected so far. In the iteration,
the algorithm selects a set $T$ of $\alpha$
nodes from $V'$ (or all
such nodes, if they are less than $\alpha$) and picks as next
center the node $v \in T$ that maximizes the number of nodes
$u \in V'$ with $\Pr(u \sim v) \geq \bar{q}$. (The role of parameters
$\alpha$ and $\bar{q}$ will be evident in the following subsections.)
At the end of the $k$ iterations,
it returns the clustering defined by the best assignment of the
covered nodes to the $k$ selected centers. It is
easy to see that the connection probability of
each covered node to its assigned cluster center is
at least $q$.
\begin{algorithm}[t]
\caption{\algo{min-partial}$(\mathcal{G},k,q,\alpha,\bar{q})$}
\label{alg:outliers}
\Let{$S$}{$\emptyset$}\Comment*{Set of centers}
\Let{$V'$}{$V$}\;
\For{\Let{$i$}{1} \To $k$} {
  select arbitrary $T \subseteq V'$
  with $|T| = \min\{\alpha,|V'|\}$\;
  \lFor{$(v \in T)$} {
    \Let{$M_v$}{$\{u \in V' \; : \; \Pr {(u \sim v)} \geq \bar{q} \}$}
  }
  \Let{$c_i$}{$\argmax_{v \in T}|M_v|$}  \nllabel{algoline:small-disk}\;
  \Let{$S$}{$S \cup \{c_i\}$}\;
  \Let{$V'$}{$V'-\{u \in V' \; : \; \Pr {(u \sim c_i)} \geq q \}$} \nllabel{algoline:large-disk}\;
}
\lIf {$(|S| < k)$}
{\\ \hspace*{0.5cm}add $k-|S|$ arbitrary nodes of $V-S$ to $S$}
\lFor{\Let{$i$}{1} \To $k$} {
\Let{$C_i$}{$\{u \in V-V' \: : \; c(u, S)=c_i \}$}}
\Return
$\mathcal{C}=(C_1,\ldots,C_k; c_1,\ldots,c_k)$\;
\end{algorithm}

\subsection{MCP clustering}
\label{sec:k-center-algo}
We now turn the attention to the MCP problem.  The following lemma
shows that if Algorithm~\algo{min-partial} is provided with a
suitable guess $q$ for the minimum connection probability, then the
returned clustering covers all nodes and is a good solution to the MCP
problem.  Let $p_{\rm opt-min}(k)$ be the maximum value of
$\mbox{min-prob}(\mathcal{C})$ over all full $k$-clusterings $\mathcal{C}$
of $\mathcal{G}$.  (Observe that $p_{\rm opt-min}(k)>0$
if and only if $\mathcal{G}$ has at most $k$ connected components and,
for convenience, we assume that this is the case.)
\begin{lemma} \label{lem:k-center}
\sloppy
For any $q \leq p^2_{\rm opt-min}(k)$, $\alpha \geq 1$,
and $\bar{q} \in [q,1]$
we have that
the $k$-clustering $\mathcal{C}$
returned by \algo{min-partial}$(\mathcal{G},k,q,\alpha,\bar{q})$
covers all nodes.
\end{lemma}
\begin{proof}
\sloppy
Consider an optimal $k$-clustering
$\hat{\mathcal{C}}=
(\hat{C}_1, \ldots, \hat{C}_k; \hat{c}_1, \ldots,
\hat{c}_k)$ of $\mathcal{G}$, with
$V=\cup_{i=1,k} \hat{C}_i$ and $\mbox{min-prob}(\hat{\mathcal{C}})
= p_{\rm opt-min}(k)$.
Let $c_i$ be the center added to $S$ in the $i$-th iteration of the
{\bf for} loop of \algo{min-partial}, and
let $\hat{C}_{j_i}$ be the cluster in $\hat{\mathcal{C}}$
which
contains $c_i$, for every $i \geq 1$.
By Theorem~\ref{triangle} we have that for every node
$v \in \hat{C}_{j_i}$
\[
\Pr{(c_i \sim v)} \geq
\Pr{(c_i \sim \hat{c}_{j_i})}
\cdot
\Pr{(\hat{c}_{j_i} \sim v)}
\ge p^2_{\rm opt-min}(k) \geq q
\]
Therefore, at the end of the $i$-th iteration of the {\bf for} loop, $V'$
cannot contain nodes of $\hat{C}_{j_i}$. An easy induction
shows that at the end of the  {\bf for} loop, $V'$ is empty.
\end{proof}
Based on the result of the lemma, we can solve the MCP problem by
repeatedly running \algo{min-partial} with progressively smaller
guesses of $q$, starting from $q=1$ and decreasing $q$ by a factor
$(1+\gamma)$, for a suitable parameter $\gamma >0$, at each run, until a
clustering covering all nodes is obtained.
We refer to this algorithm as \algo{MCP}
(Algorithm~\ref{alg:k-center} in the box).
\begin{algorithm}[t]
\caption{\algo{MCP}$(\mathcal{G},k,\gamma)$}
\label{alg:k-center}
\Let{$q$}{1}\;
\While{true} {
\Let{$\mathcal{C}$}{\algo{min-partial}$(\mathcal{G},k,q,1,q)$}\;
\lIf {$\mathcal{C}$ covers all nodes}{\Return $\mathcal{C}$}
\lElse {\Let{$q$}{$q/(1+\gamma)$}}
}
\end{algorithm}
The following theorem is an immediate consequence of
Lemma~\ref{lem:k-center}.
\begin{theorem}~\label{thm:k-center}
Algorithm~\ref{alg:k-center} requires
at most \mbox{$\lfloor 2\log_{1+\gamma}(1/p_{\rm opt-min}(k))\rfloor+1$}
executions of \algo{min-partial}, and
returns a $k$-clustering $\mathcal{C}$ with
\[
\mbox{min-prob}(\mathcal{C}) \geq \frac{p^2_{\rm opt-min}(k)}{(1+\gamma)}.
\]
\end{theorem}
It is easy to see that all connection probabilities
 $\Pr(u \sim v)$ used in  Algorithm~\ref{alg:k-center}
are not smaller than $p^2_{\rm opt-min}(k)/(1+\gamma)$.  Also, we observe that
once $q$ becomes sufficiently small to ensure the existence of a full
$k$-clustering, a binary search between the last two
guesses for $q$ can be performed to get a higher minimum connection
probability.

\subsection{ACP clustering}
\label{sec:k-median-algo}
In order to compute good solutions to the ACP problem we resort again
to the computation of partial clusterings. For a given connection probability
threshold $q \in (0,1]$, define $t_q$ as the minimum number of nodes
left uncovered by any partial $k$-clustering $\mathcal{C}$
of $\mathcal{G}$ with $\mbox{min-prob}(\mathcal{C}) \geq q$.
It is easy to argue that $t_q$ is a non-decreasing function of $q$.
Observe that any partial $k$-clustering can be ``completed", i.e.,
turned into a full $k$-clustering, by assigning the
uncovered nodes arbitrarily to the available clusters
(possibly with connection probabilities to the cluster centers
equal to 0), and that
$q (n-t_q)/n$ is a lower bound to the average connection probability
of such a full $k$-clustering.
Let $p_{\rm opt-avg}(k)$ be the maximum value of
$\mbox{avg-prob}(\mathcal{C})$ over all $k$-clusterings $\mathcal{C}$
of $\mathcal{G}$. The following lemma shows that
for a suitable $q$, the value $q (n-t_q)/n$ is not much smaller than
$p_{\rm opt-avg}(k)$.
\begin{lemma} \label{lem:k-median}
There exists a value $q \in (0,1]$ such that
\[
q \cdot {n-t_q \over n} \geq {p_{\rm opt-avg}(k) \over H(n)},
\]
where $H(n) = \sum_{i=1}^{n}(1/i)=\ln n  + O(1)$ is the $n$-th harmonic number.
\end{lemma}
\begin{proof}
Let  $\hat{\mathcal{C}}$ be the $k$-clustering of $\mathcal{G}$
which maximizes the average connection probability.
Let $p_0 \leq p_1 \leq \cdots \leq p_{n-1}$ the connection probabilities
of the $n$ nodes to their cluster centers in  $\hat{\mathcal{C}}$,
sorted by non-decreasing order, and note that
$\mbox{avg-prob}(\hat{\mathcal{C}}) = (1/n) \sum_{i=0}^{n-1} p_i$.
It is easy to argue that for each $0 \leq i < n$ there exists
a partial $k$-clustering of $\mathcal{G}$ which covers $n-i$ nodes and where
each covered node is connected to its cluster center with probability
at least $p_i$. This implies that $t_{p_i} \leq i$. We claim that
\[
\max_{i=0, \ldots, n-1} p_i {n-i \over n} \geq {p_{\rm opt-avg}(k) \over H(n)}.
\]
If this were not the case we would have
\[
p_{\rm opt-avg}(k) =
{1 \over n} \sum_{i=0}^{n-1}  p_i
<
{p_{\rm opt-avg}(k) \over H(n)} \sum_{i=0}^{n-1} {1 \over n-i} = p_{\rm opt-avg}(k),
\]
which is impossible. Therefore, there must exist an index $i \in [0,n-1]$
such that
\[
  p_i {n-t_{p_i} \over n} \geq p_i {n-i\over n} \geq {p_{\rm opt-avg}(k) \over H(n)}.
  \qquad\qedhere
\]
\end{proof}
Based on the above lemma, we can obtain an approximate solution
for the ACP problem by seeking partial $k$-clusterings
which strike good tradeoffs between the minimum connection
probability and the number of uncovered points. The next lemma shows that
Algorithm \algo{min-partial} can indeed provide these partial clusterings.
\begin{lemma} \label{lem:outliers}
For any $q \in (0,1]$, we have that
the partial $k$-clustering $\mathcal{C}$
returned by \algo{min-partial}$(\mathcal{G},k,q^3,n,q)$
covers all but at most $t_q$ nodes.
\end{lemma}

\iflong 
\begin{proof}
  The proof can be obtained by rephrasing the proof of the 3-approximation result for the \emph{robust $k$-center problem} in \cite[Theorem 3.1]{CharikarKMN01} in terms of connection probabilities rather than distances, as follows.

  Let $\{o_1, \dots, o_k\} \subseteq V$ be the set of centers of a partial $k$-clustering $\mathcal{C}$ of $\mathcal{G}$ with $\mbox{min-prob}(\mathcal{C}) \ge q$ minimizing the number of uncovered nodes $t_q$.
  Let also $O_1, \dots, O_k$ be the \emph{optimal disks} of radius $q$ centered at $o_1, \dots, o_k$, that is $O_i = \{v : \Pr(v \sim o_i) \ge q\}$.
  Similarly, consider the \emph{greedy disks} of radius $q$ selected on line \ref{algoline:small-disk} of Algorithm~\ref{alg:outliers}, and those of radius $q^3$ subtracted from $V'$ on line \ref{algoline:large-disk}, and call them $M_1, \dots, M_k$ and $E_1, \dots, E_k$, respectively. Let $c_1, \dots, c_k$ be the centers of such disks (note that disks with the same index share the same center).
  To prove the lemma we show that $E_1, \dots, E_k$ cover at least the same number of nodes covered by $O_1, \dots, O_k$, that is
  \[
    |E_1 \cup \dots \cup E_k| \ge |O_1 \cup \dots \cup O_k|
  \]
  In order to do so, we show that there is a permutation of $O_1, \dots, O_k$, say  $O_{\pi(1)}, \dots, O_{\pi(k)}$ such that, for $i \in [1, k]$:
  \[
    |E_1 \cup \dots \cup E_i| \ge |O_{\pi(1)} \cup \dots \cup O_{\pi(i)}|.
  \]
  The proof proceeds by constructing the permutation inductively,  using a charging argument in which we associate each node of the optimal disks to a distinct point of disks $E_1, \dots, E_i$.
  Assume that the induction hypothesis holds for $i-1$. We have two cases.
  \begin{enumerate}
  \item $(M_1 \cup \dots \cup M_i) \cap O_j \neq \emptyset$ for some $j\in [1,k] \setminus \{\pi(1), \ldots, \pi(i-1)\}$.
    In this case we let $\pi(i) = j$.
   Observe that,  by Theorem~\ref{triangle}, there is at least one center $c_\ell$ for $\ell \in [1, i]$ such that $\Pr(c_\ell \sim v) \ge q^3$ for any $v \in O_{\pi(i)}$.
    This implies that $O_{\pi(i)} \subseteq E_1 \cup \dots \cup E_i$, thus we can charge each uncharged point of $O_{\pi(i)}$ to itself.
    We shall see that it is impossible for these points to have been matched with some other point in a previous iteration.
  \item $(M_1 \cup \dots \cup M_i) \cap O_j = \emptyset, \forall j\in [1,k]  \setminus \{\pi(1), \ldots, \pi(i-1)\}$.
    Let $U = V \setminus (E_1 \cup \dots \cup E_{i-1})$ and set $\pi(i)$ to the index of the optimal disk maximizing $O_j \cap U$, $\forall j\in [1,k] \setminus \{\pi(1), \ldots, \pi(i-1)\}$. By the greedy choice we have $|M_i \cap U| \ge |O_{\pi(i)} \cap U|$, otherwise we would have selected $O_{\pi(i)}$ in place of $M_i$ in iteration $i$. Therefore there are enough points to charge each point of $O_{\pi(i)}$ to a distinct point of $M_i$.
    No remaining optimal disk will attempt to charge these same points to themselves (as an application of Case 1) because $M_i$ is disjoint from any such optimal disk.
    Since $M_i \subseteq E_i$ we have that all points in $O_{\pi(1)} \cup \dots \cup O_{\pi(i)}$ are charged to a distinct point of $E_1 \cup \dots \cup E_i$.
  \end{enumerate}
  This proves that $|E_1 \cup \dots \cup E_k| \ge |O_1 \cup \dots \cup O_k|$, and the lemma follows.
\end{proof}

\fi

We are now ready to describe our approximation
algorithm for the ACP problem, which we refer to as \algo{ACP}. The idea is to run
\algo{min-partial} so to obtain partial $k$-clusterings with $t_q$ uncovered
nodes, for progressively smaller values of $q$.
For each such value of $q$, \algo{min-partial} is employed to compute
a partial $k$-clustering $\mathcal{C}$ where at most $t_q$
nodes remain uncovered and where  each covered node is
connected to its cluster center with probability at least
$q^3$. If $\mathcal{C}$ has the potential to provide
a full $k$-clustering with higher average connection probability
than those previously found, it is turned into a full
  $k$-clustering  by assigning
the uncovered nodes to arbitrary clusters.
The algorithm halts when further smaller guesses of $q$
cannot lead to better clusterings.
The pseudocode (Algorithm~\ref{alg:k-median} in the box) uses the following
notation. For a (partial) $k$-clustering $\mathcal{C}$ and a node
$u \in V$, $p_\mathcal{C}(u)$ denotes the connection
probability of $u$ to the center of its assigned cluster, if any, and set  $p_\mathcal{C}(u)=0$
otherwise.
\begin{algorithm}[t]
\caption{\algo{ACP}$(\mathcal{G},k,\gamma)$}
\label{alg:k-median}
\Let{$\mathcal{C}$}{\algo{min-partial}$(\mathcal{G},k,1,n,1)$}\;
\Let{$\phi_{\mbox{best}}$}{$(1/n)\sum_{u \in V} p_\mathcal{C}(u)$}\;
\Let{$\mathcal{C}_{\mbox{best}}$}{any full $k$-clustering completing $\mathcal{C}$}\;
\Let{$q$}{$1/(1+\gamma)$}\;
\While{$(q^3 \geq \phi_{\mbox{\em best}})$} {
\Let{$\mathcal{C}$}{\algo{min-partial}$(\mathcal{G},k,q^3,n,q)$}\;
\Let{$\phi$}{$(1/n)\sum_{u \in V} p_\mathcal{C}(u)$}\;
\If {$(\phi \geq \phi_{\mbox{\em best}})$}{
\Let{$\phi_{\mbox{best}}$}{$\phi$} \;
\Let{$\mathcal{C}_{\mbox{best}}$}{any full $k$-clustering completing $\mathcal{C}$}\;
}
\lElse {\Let{$q$}{$q/(1+\gamma)$}}
}
\Return $\mathcal{C}_{\mbox{best}}$
\end{algorithm}
We have:
\begin{theorem}~\label{thm:k-median}
Algorithm~\ref{alg:k-median}
returns a $k$-clustering $\mathcal{C}$ with
\[
\mbox{avg-prob}(\mathcal{C}) \geq
\left(\frac{p_{\rm opt-avg}(k)}{(1+\gamma) H(n)}\right)^3,
\]
where $H(n)$ is the $n$-th harmonic number, and
requires at most $\left\lfloor\log_{1+\gamma}(H(n)/p_{\rm opt-avg}(k)) \right\rfloor + 1$ executions of \algo{min-partial}.
\end{theorem}
\begin{proof}
Note that the while loop maintains, as an  invariant, the
relation
$\mbox{avg-prob}(\mathcal{C}_{\mbox{best}}) \geq  \phi_{\mbox{best}}$.
Hence, this relations holds at the end of the algorithm when
the $k$-clustering $\mathcal{C}_{\mbox{best}}$ is returned.
Let $q^* \in (0,1]$ be a value such that
\[
q^* \cdot {n-t_{q^*} \over n} \geq {p_{\rm opt-avg}(k) \over H(n)}.
\]
The existence of $q^*$ is ensured by Lemma~\ref{lem:k-median}.
If the while loop ends when $q > q^*$, then
\begin{eqnarray*}
\phi_{\mbox{best}} & > & q^3
> (q^*)^3
\geq
\left(
{n \over n-t_{q^*}}
{p_{\rm opt-avg}(k) \over H(n)}
\right)^3 \\
& \geq &
\left(
{p_{\rm opt-avg}(k) \over H(n)}
\right)^3.
\end{eqnarray*}
If instead $q$ becomes $\leq q^*$,
consider the first iteration of the while loop when this happens, that is when
$q^*/(1+\gamma) < q \leq q^*$ and let $\mathcal{C}$
be the partial $k$-clustering computed in the iteration. By
Lemma~\ref{lem:outliers}, at most $t_q$ nodes are not
covered by $\mathcal{C}$ and since $t_q$ is non-decreasing, as observed before,
we have that $t_q < t_{q^*}$. This implies that
the value $\phi$ derived from $\mathcal{C}$ (hence, $\phi_{\mbox{best}}$
at the end of the iteration) is such that
\begin{eqnarray*}
\phi & > & q^3 \cdot {n-t_q \over n}
\quad \geq \quad
 \left({q^* \over 1+\gamma}\right)^3 \cdot {n-t_{q^*} \over n} \\
&\geq &
\left(
{n \over n-t_{q^*}}
{p_{\rm opt-avg}(k) \over (1+\gamma) H(n)}
\right)^3 \cdot {n-t_{q^*} \over n}  \\
& \geq &
\left(
{p_{\rm opt-avg}(k) \over (1+\gamma) H(n)}
\right)^3.
\end{eqnarray*}
In all cases, the average connection probability of the returned clustering satisfies the
stated bound. As for the upper bound on the number of iterations of the while
loop, we proved above that as soon as $q$ falls in the interval
$(q^*/(1+\gamma), q^*]$ we have
$\phi_{\mbox{best}} \geq (p_{\rm opt-avg}(k)/((1+\gamma) H(n)))^3$,
hence, from that point on, $q$ cannot become smaller than
$p_{\rm opt-avg}(k)/((1+\gamma) H(n)) < q^*$.
This implies that $\left\lfloor\log_{1+\gamma}(H(n)/p_{\rm opt-avg}(k))\right\rfloor + 1$
iterations of the while loop are executed overall.
\end{proof}
It is easy to see that all connection probabilities
 $\Pr(u \sim v)$ that Algorithm~\ref{alg:k-median}
needs to compute in order to be correct are
not smaller than $(p_{\rm opt-avg}(k)/((1+\gamma) H(n)))^3$.

We remark that while the theoretical approximation ratios attained by
our algorithms for both the MCP and ACP problems appear somewhat weak,
especially for small values of $p_{\rm opt-min}$ and $p_{\rm
  opt-avg}$, we will provide experimental evidence (see
Section~\ref{sec:experiments}) that, in practical scenarios where
connection probabilities are not too small, they return good-quality
clusterings and, by avoiding the estimation of small connection
probabilities, they run relatively fast.

\subsection{Limiting the path length}
\label{sec-depthlimit}

\newcommand{\simd}{\stackrel{d}{\sim}}
\newcommand{\simdPar}[1]{\stackrel{#1}{\sim}}
The algorithms described in the preceding subsections can be
run by setting a limit on the length of the paths that contribute to
the connection probability between two nodes. As mentioned in the
introduction, this feature may be useful in application scenarios where
the similarity between two nodes
diminishes steeply with their topological distance regardless
of their connection probability.

For a fixed integer $d$, with $1 \leq d < n$, we define
$\Pr(u\simd v) =   \sum_{G\sqsubseteq \mathcal{G}} \Pr(G) \textbf{I}_{G}(u, v;d)$,
where $\textbf{I}_{G}(u, v;d)$ is 1 if $u$ is at distance at most $d$ from $v$ in $G$,
and 0 otherwise. In the following, we refer to $\Pr(u\simd v)$ as the
\emph{$d$-connection probability} between $u$ and $v$. By easily adapting the
proofs of Lemma~\ref{conditional1} and Theorem~\ref{triangle}, it can be
shown that for any pair of distances $d_1, d_2$, with $d \geq d_1+d_2$, and
for any triplet $u,v,z \in V$, it holds that
\begin{equation}
  \label{eq:triangle-d-connection}
\Pr(u \simd z) \geq \Pr(u \simdPar{d_1} v) \cdot
\Pr(v \simdPar{d_2} z).
\end{equation}
We now reconsider the MCP and ACP problems under paths of limited depth.
For a $k$-clustering $C_1, \dots, C_k$ with centers
$\{c_1, \dots, c_k\}$, we define the objective functions
\begin{align}
\mbox{min-prob}_d(\mathcal{C}) & =
\min_{1\le i \le k} \min_{v \in C_i}  \Pr(c_i \simd v) \label{d-pmin}, \\
\mbox{avg-prob}_d(\mathcal{C}) & =
(1/n) \sum_{1\le i \le k} \sum_{v \in C_i} \Pr(c_i \simd v), \label{d-pmedian}
\end{align}
\sloppy and let $p_{\rm opt-min}(k,d)$ and $p_{\rm opt-avg}(k,d)$, respectively,
be the maximum values of these two objective
functions over all $k$-clusterings.

\iflong 
Suppose that we modify
Algorithm~\ref{alg:outliers} to employ $d$-connection probabilities
rather than the unconstrained connection probabilities, as detailed in Algorithm~\ref{alg:min-partial-d-connection}.
We introduce two new parameters, namely $d$ and $d'$, with $d \ge d'$:
disks built on line~\ref{algoline:small-disk-limited-depth} are defined in terms of $d'$-connection probabilities, whereas disks built on line~\ref{algoline:large-disk-limited-depth} consider $d$-connection probabilities.
We call this variant \algo{min-partial-d}.

We consider the \algo{mcp} problem first. The following lemma is analogous to Lemma~\ref{lem:k-center}.
\begin{lemma} \label{lem:k-center-limited-d}
For any $q \leq p^2_{\textrm{opt-min}}(k, \lfloor d/2 \rfloor)$, $\alpha \geq 1$,
and $\bar{q} \in [q,1]$
we have that
the $k$-clustering $\mathcal{C}$
returned by \algo{min-partial-d}$(\mathcal{G},k,q,\alpha,\bar{q}, d, d)$
covers all nodes.
\end{lemma}
\begin{proof}
Consider an optimal $k$-clustering
$\hat{\mathcal{C}}=
(\hat{C}_1, \ldots, \hat{C}_k; \hat{c}_1, \ldots,
\hat{c}_k)$ of $\mathcal{G}$, with
$V=\cup_{i=1,k} \hat{C}_i$ and $\mbox{min-prob}(\hat{\mathcal{C}})
= p_{\rm opt-min}(k, \lfloor d/2 \rfloor)$.
Let $c_i$ be the center added to $S$ in the $i$-th iteration of the
{\bf for} loop of \algo{min-partial}, and
let $\hat{C}_{j_i}$ be the cluster in $\hat{\mathcal{C}}$
which
contains $c_i$, for every $i \geq 1$.
By Inequality~\ref{eq:triangle-d-connection} we have that for every node
$v \in \hat{C}_{j_i}$
\[
  \begin{aligned}
    \Pr{(c_i \simd v)} \geq&
    \Pr{(c_i \simdPar{\lfloor d/2 \rfloor} \hat{c}_{j_i})}
    \cdot
    \Pr{(\hat{c}_{j_i} \simdPar{\lfloor d/2 \rfloor} v)} \\
    \ge& p^2_{\textrm{opt-min}}(k, \lfloor d/2 \rfloor) \geq q
  \end{aligned}
\]
Therefore, at the end of the $i$-th iteration of the {\bf for} loop, $V'$
cannot contain nodes of $\hat{C}_{j_i}$. An easy induction
shows that at the end of the  {\bf for} loop, $V'$ is empty.
\end{proof}

To run Algorithm~\ref{alg:k-center} with $d$-connection probabilities, we just need to replace the invocation of \algo{min-partial} with an invocation to \algo{min-partial-d}$(\mathcal{G},k,q,\alpha,\bar{q}, d, d)$.
The following Theorem is then a direct consequence of the above lemma. 

\begin{theorem} \label{thm:d-kcenter}
Suppose that $p_{\text{opt-min}}(k, \lfloor d/2 \rfloor)>0$.
When run with $d$-connection probabilities,
Algorithm~\ref{alg:k-center} requires
at most \mbox{$\lfloor 2\log_{1+\gamma}(1/p_{\rm opt-min}(k,\lfloor d/2 \rfloor))\rfloor+1$}
executions of \algo{min-partial-d}, and
returns a $k$-clustering $\mathcal{C}$ with
\[
\mbox{min-prob}_d(\mathcal{C}) \geq \frac{p^2_{\rm opt-min}(k,\lfloor d/2 \rfloor)}{(1+\gamma)}.
\]
\end{theorem}
We remark that the assumption $p_{\text{opt-min}}(k, \lfloor d/2 \rfloor)>0$ is needed to ensure that a suitable guess of $q$ is reached in a finite number of iterations.

\begin{algorithm}[t]
  \caption{\algo{min-partial-d}$(\mathcal{G}, k, q, \alpha, \bar{q}, d, d')$  adaptation for depth-limited algorithms}
  \label{alg:min-partial-d-connection}
  \Let{$S$}{$\emptyset$}\;
  \Let{$V'$}{$V$}\;
  \For{\Let{$i$}{1} \To $k$} {
    select arbitrary $T \subseteq V'$
    with $|T| = \min\{\alpha,|V'|\}$\;
    \lFor{$(v \in T)$} {
      \nllabel{algoline:small-disk-limited-depth}
      \Let{$M_v$}{$\{u \in V' \; : \; \Pr {(u \simdPar{d'} v)} \geq \bar{q} \}$}
    }
    \Let{$c_i$}{$\argmax_{v \in T}|M_v|$}\;
    \Let{$S$}{$S \cup \{c_i\}$}\;
    \Let{$V'$}{$V'-\{u \in V' \; : \; \Pr {(u \simd c_i)} \geq q \}$}\nllabel{algoline:large-disk-limited-depth}\;
  }
  \lIf {$(|S| < k)$}
  {\\ \hspace*{0.5cm}add $k-|S|$ arbitrary nodes of $V-S$ to $S$}
  \lFor{\Let{$i$}{1} \To $k$} {
    \Let{$C_i$}{$\{u \in V-V' \: : \; c(u, S)=c_i \}$}}
  \Return
  $\mathcal{C}=(C_1,\ldots,C_k; c_1,\ldots,c_k)$\;
\end{algorithm}

Consider now the \algo{acp} problem.
For a given probability threshold $q$, and a given $d>0$, define $t_{q, d}$ as the minimum number of nodes left uncovered by any partial $k$-clustering $\mathcal{C}$ of $\mathcal{G}$ with $\mbox{min-prob}_d(\mathcal{C}) \ge q$.
In the following, $t_{q, d}$ plays a role similar to $t_q$ in Section~\ref{sec:k-median-algo}.

Specifically, we can state the following Lemma, equivalent to Lemma~\ref{lem:k-median}.

\begin{lemma}\label{lem:k-median-limited-depth}
  For any given $d>0$, There exist a value $q \in (0, 1]$ such that
  \[
    q\cdot \frac{n - t_{q, d}}{n}\ge
    \frac{p_{\textrm{opt-avg}}(k, d)}{H(n)}
  \]
  where $H(n) = \sum_{i=1}^n (1/i) = \ln n + \BO{1}$ is the $n$-th harmonic number.
\end{lemma}
\begin{proof}
  The proof follows by the same argument of the proof of Lemma~\ref{lem:k-median}, by considering $d$-connection probabilities anywhere unconstrained connection probabilities are used.
\end{proof}

The following lemma ensures that, for a given probability threshold, the number of uncovered nodes after an invocation of \algo{min-partial-d} is conveniently bounded.

\begin{lemma} \label{lem:outliers-limited-depth}
  For any $d>0$, and for any $q \in (0,1]$, we have that
  the partial $k$-clustering $\mathcal{C}$
  returned by \algo{min-partial-d}$(\mathcal{G},k,q^3,n,q,d,\lfloor d/3 \rfloor)$
  covers all but $t_{q, \lfloor d/3 \rfloor}$ nodes.
\end{lemma}
\begin{proof}
  The proof of this lemma proceeds as the one of Lemma~\ref{lem:outliers}, with the  following modification of Case (1) of the induction step.
  The hypothesis of Case (1) is that $(M_1 \cup \dots \cup M_i) \cap O_j \neq \emptyset$ for some $j\in [1,k]\setminus\{\pi(1), \ldots, \pi(i-1)\}$.
 We then set $\pi(i) = j$. Let now $v \in O_{\pi(i)}$. Since there is at least one center $c_\ell$, with $\ell \in [1, i]$, such that $M_\ell \cap O_{\pi(i)} \neq \emptyset$, we have that 
  \[
    \Pr(c_\ell \simd v) \ge
    \Pr(c_\ell \simdPar{\lfloor d/3 \rfloor} x) \cdot
    \Pr(x \simdPar{\lfloor d/3 \rfloor} o_i) \cdot
    \Pr(o_i \simdPar{\lfloor d/3 \rfloor} v)
    \ge q^3
  \]
  where $x \in M_\ell \cap O_{\pi(i)}$.
  The first inequality comes from Equation~(\ref{eq:triangle-d-connection}), and the second by construction (line~\ref{algoline:small-disk-limited-depth} of Algorithm~\ref{alg:min-partial-d-connection}).
  Therefore we have that the $d$-connection probability between $c_\ell$ and any $v$ in $O_{\pi(i)}$ is greater than $q^3$, which means that we can charge every point of $O_{\pi(i)}$ to itself.
  The rest of the argument is unvaried, and the theorem follows.
\end{proof}

The adaptation of Algorithm~\ref{alg:k-median} for the \algo{acp} problem to work with $d$-connection probabilities consists in replacing the invocation to \algo{min-partial} with an invocation to \algo{min-partial-d}$(\mathcal{G},k,q^3,n,q,d, \lfloor d/3 \rfloor)$, obtaining the following theorem.

\begin{theorem}\label{thm:d-kmedian}
  When run with $d$-connection probabilities, Algorithm~\ref{alg:k-median}
  returns a $k$-clustering $\mathcal{C}$ with
  \[
    \mbox{avg-prob}_d(\mathcal{C}) \geq
    \left(\frac{p_{\rm opt-avg}(k,\lfloor d/3 \rfloor)}{(1+\gamma) H(n)}\right)^3,
  \]
  where $H(n)$ is the $n$-th harmonic number, and requires at most
  $\left\lfloor\log_{1+\gamma}(H(n)/p_{\rm opt-avg}(k,\lfloor d/3 \rfloor)) \right\rfloor + 1$
  executions of \algo{min-partial}.
\end{theorem}
\begin{proof}
  This proof is structured as the proof of Theorem~\ref{thm:k-median}, with the caveat that, for a given $q$, the clustering returned by Algorithm~\ref{alg:min-partial-d-connection} has radius $q^3$ in terms of $d$-connection probability, whereas the number of uncovered nodes is limited by $t_{q, \lfloor d/3 \rfloor}$.
  
  Note that the while loop of Algorithm~\ref{alg:k-median} maintains, as an  invariant, the relation
  $\mbox{avg-prob}_d(\mathcal{C}_{\mbox{best}}) \geq  \phi_{\mbox{best}}$, where $\phi_{\mbox{best}}$ is defined using $d$-connection probabilities.
  Hence, this relation holds at the end of the algorithm when the $k$-clustering $\mathcal{C}_{\mbox{best}}$ is returned.
  Let $q^* \in (0,1]$ be a value such that
  \[
    q^* \cdot {n-t_{q^*, \lfloor d/3 \rfloor} \over n} \geq {p_{\textrm{opt-avg}}(k, \lfloor d/3 \rfloor) \over H(n)}.
  \]
  The existence of $q^*$ is ensured by Lemma~\ref{lem:k-median-limited-depth}.
  If the while loop ends when $q > q^*$, then
  \begin{eqnarray*}
    \phi_{\mbox{best}}
    & > & q^3 > (q^*)^3
      \geq
      \left(
      {n \over n-t_{q^*}}
      {p_{\rm opt-avg}(k, \lfloor d/3 \rfloor) \over H(n)}
      \right)^3 \\
    & \geq &
      \left(
      {p_{\rm opt-avg}(k, \lfloor d/3 \rfloor) \over H(n)}
      \right)^3.
  \end{eqnarray*}
  If instead $q$ becomes $\leq q^*$,
  consider the first iteration of the while loop when this happens, that is when
  $q^*/(1+\gamma) < q \leq q^*$ and let $\mathcal{C}$
  be the partial $k$-clustering computed in the iteration. By
  Lemma~\ref{lem:outliers-limited-depth}, at most $t_{q, \lfloor d/3 \rfloor}$ nodes are not
  covered by $\mathcal{C}$ and since $t_{q, \lfloor d/3 \rfloor}$ is non-decreasing, as observed before, we have that $t_{q, \lfloor d/3 \rfloor} < t_{q^*, \lfloor d/3 \rfloor}$. This implies that
  the value $\phi$ derived from $\mathcal{C}$ (hence, $\phi_{\mbox{best}}$
  at the end of the iteration) is such that
  \begin{eqnarray*}
    \phi & > & q^3 \cdot {n-t_{q, \lfloor d/3 \rfloor} \over n}
        \quad \geq \quad
        \left({q^* \over 1+\gamma}\right)^3 \cdot {n-t_{q^*, \lfloor d/3 \rfloor} \over n} \\
      &\geq &
        \left(
        {n \over n-t_{q^*, \lfloor d/3 \rfloor}}
        {p_{\rm opt-avg}(k, \lfloor d/3 \rfloor) \over (1+\gamma) H(n)}
        \right)^3 \cdot {n-t_{q^*, \lfloor d/3 \rfloor} \over n}  \\
      & \geq &
        \left(
        {p_{\rm opt-avg}(k, \lfloor d/3 \rfloor) \over (1+\gamma) H(n)}
        \right)^3.
  \end{eqnarray*}
  In all cases, the average connection probability of the returned clustering satisfies the
  stated bound. As for the upper bound on the number of iterations of the while
  loop, we proved above that as soon as $q$ falls in the interval
  $(q^*/(1+\gamma), q^*]$ we have
  $\phi_{\mbox{best}} \geq (p_{\rm opt-avg}(k, \lfloor d/3 \rfloor)/((1+\gamma) H(n)))^3$,
  hence, from that point on, $q$ cannot become smaller than
  $p_{\rm opt-avg}(k, \lfloor d/3 \rfloor)/((1+\gamma) H(n)) < q^*$.
  This implies that $\left\lfloor\log_{1+\gamma}(H(n)/p_{\rm opt-avg}(k, \lfloor d/3 \rfloor))\right\rfloor + 1$
  iterations of the while loop are executed overall.
  Termination is always guaranteed since $p_{\rm opt-avg}(k,\lfloor d/3 \rfloor) \geq k/n>0$, for every $d \geq 0$.
\end{proof}

\fi


\section{Implementing the oracle}
\label{sec:implementation}

In the previous section we assumed that the probabilities $\Pr(u \sim
v)$ could be obtained exactly from an oracle. In practice, the
estimation of these probabilities is the most critical part for the
efficient implementation of our algorithms. In this section, we show
how to integrate the Monte Carlo sampling method for the estimation
of the connection probabilities within the algorithms described in
Section~\ref{sec:algs}, maintaining similar guarantees on the quality
of the returned clusterings. The basic idea of our approach is
to adjust the number of samples dynamically during the execution of
the algorithms, based on safe guesses of the probabilities that
need to be estimated.

For ease of presentation, throughout this section we assume that lower
bounds to $p^2_{\rm opt-min}(k)$ and to $(p_{\rm opt-avg}(k)/H(n))^3$
are available. We will denote both lower bounds by $p_L$, since it will be
clear from the context which one is used. 
For example, these lower bounds can be obtained by observing that
$p_{\rm opt-min}(k)$ is greater than or equal to the probability of the
most unlikely world, and $p_{\rm opt-avg}(k) \geq k/n$. In practice,
$p_L$ can be employed as a threshold set by
the user to exclude a priori clusterings with low values of the
objective function. In this case, if the algorithm does
not find a clustering whose objective function is above the threshold, it
terminates by reporting that no clustering could be found. 
Recall that $\tilde{p}(u,v)$ denotes the estimate of the probability
$\Pr(u \sim v)$ obtained by sampling possible worlds (see
Equation~\ref{eq:estimator}). Moreover, for a node $u \in V$ and a set
of nodes $S \subset V$, we define $\tilde{c}(u, S) = \argmax_{c\in S}
\{ \tilde{p}(c, u) \}$ as the function returning the node of $S$
connected to $u$ with the highest \emph{estimated} probability.
Similarly, we define $\tilde{\pi}(u,S)=\Pr(\tilde{c}(u,S) \sim u) = \max_{c\in S}
\{ \tilde{p}(c, u) \}$. We use $\epsilon > 0$ to denote an approximation
parameter to be fixed by the user.

\subsection{Partial clustering}
A key component of Algorithms~\algo{mcp} and
\algo{acp} is the \algo{min-partial} subroutine
(Algorithm~\ref{alg:outliers}), whose input includes two thresholds
$q$ and $\bar{q}$ for the connection probabilities.  Note that in each
invocation of \algo{min-partial} within \algo{mcp}
and \algo{acp}, we have that $\bar{q} \geq q$, and that
only connection probabilities not smaller than $q$ are needed. 
We implement \algo{min-partial} as follows. 
Suppose we estimate connection probabilities
using a number $r$ of samples, 
based on Equation~(\ref{eq:num-samples}), 
which ensures that any $\Pr(u \sim v) \geq q$
is estimated with relative error
at most $\epsilon/2$ with probability at least $1-\delta$,
where $\epsilon, \delta \in (0,1)$ are suitable
values. Then, in each of the $k$
iterations of the main for-loop of \algo{min-partial}, a new center
$c$ is selected which maximizes the number of uncovered nodes $u$ with
$\tilde{p}(c, u) \ge (1 - \frac{\epsilon}{2})\bar{q}$, and all nodes
$u$ with $\tilde{p}(c, u) \ge (1 - \frac{\epsilon}{2})q$ are removed
from the set $V'$ of uncovered nodes. The following two subsections
analyze the quality of the clusterings returned by
Algorithms~\algo{mcp} and \algo{acp} when using this
implementation of \algo{min-partial}.

\subsection{Implementation of MCP}
\label{sec:k-center-impl}
Recall that Algorithm~\algo{mcp} 
invokes \algo{min-partial} with a probability threshold
$q$ which is lowered at each iteration 
of its main while loop.
Using the implementation of \algo{min-partial}
described before, this iterative adjustment of $q$
corresponds to a progressive sampling strategy.
In particular, if for each iteration of the while loop we use
a number of samples 
\begin{equation}\label{eq:sequential-samples}
r = 
\left\lceil
\frac{12}{q \varepsilon^2} \ln \left(2 n^3\left(1+\left\lfloor \log_{1+\gamma} \frac{1}{p_L} \right\rfloor \right) \right)
\right\rceil,
\end{equation}
we obtain the following result.
\begin{theorem}
\label{thm:seq1}
The implementation of \algo{mcp} terminates after at most
\mbox{$\lfloor 2\log_{1+\gamma}(1/p_{\rm opt-min}(k))\rfloor+1$} iterations of the while loop and returns a clustering $\tilde{\mathcal{C}}$ with
\[
\text{min-prob}(\tilde{\mathcal{C}}) \geq
\frac{(1-\epsilon)}{(1+\gamma)} p^2_{\rm opt-min}(k)
\]
if $p^2_{\rm opt-min}(k) > p_L$, with high probability.
\end{theorem}
\begin{proof}
Consider an arbitrary iteration of the while loop of
\algo{mcp}, and a pair of nodes
$u, v \in V$.  
For
\[
\delta = 
\frac{1}{n^3 \left(1+\left\lfloor\log_{1+\gamma} \frac{1}{p_L}\right\rfloor\right)},
\]
we have that using the number of samples specified by
Equation~(\ref{eq:sequential-samples}),
the following properties for the estimate $\tilde{p}(u,v)$ hold:
\begin{itemize}
\item if $\Pr(u \sim v) > q$, then $\tilde{p}(u, v) < (1-\frac{\epsilon}{2})q$ with probability $< \delta$.
\item if $\Pr(u \sim v) < (1-\epsilon) q$, then $\tilde{p}(u, v) \ge
  (1-\frac{\epsilon}{2})q$ with probability $< \delta$.
\end{itemize}
Moreover, note that \algo{mcp} performs at most $1 + \lfloor
\log_{1+\gamma} \frac{1}{p_L} \rfloor$ iterations of the while
loop. Therefore, by union bound on the
number of node pairs and the number of iterations, we have that the
following holds with probability at least $1-1/n$: in each iteration
of the while loop  every node connected to some center with probability
$\ge q$ is added to a cluster, and no cluster contains nodes whose
connection probability to the center is $< (1-\epsilon)q$. 
Consider now the $\ell$-th iteration, with
$\ell = \lfloor 2 \log_{1+\gamma}(1/p_{\rm opt-min}(k))\rfloor+1$,
in which we have $q \le p_{\rm opt-min}^2(k)$.
In this iteration, the algorithm
completes and returns a clustering.
Since at the beginning of this iteration we have
\mbox{$q > p^2_{\rm opt-min}(k)/(1+\gamma)$},
by  the above discussion we have that
  \[
    \text{min-prob}(\tilde{\mathcal{C}}) \ge (1-\epsilon)q \ge
    \frac{(1-\epsilon)}{(1+\gamma)} p^2_{\rm opt-min}(k)
  \]
with probability at least $1-1/n$, and the theorem follows
\end{proof}
We observe that, for constant $\gamma$, the overall number of
samples required by our implementation is $\BO{(1/(p_{\rm
    opt-min}(k)\epsilon)^2) (\log n+\log\log (1/p_L))}$.  As we
mentioned before, $p_L$ can be safely set equal to the probability of
the most unlikely world. In case this lower bound were too small, a
progressive sampling schedule similar to the one adopted in
\cite{pietracaprina2010mining} could be used where $p_L$ is not
required, which enures that $\BO{(1/(p_{\rm opt-min}(k)\epsilon)^2)
  (\log n+\log (1/p_{\rm opt-min}(k))}$  samples suffice. More
details will be provided in the full version of the paper.

\subsection{Implementation of ACP}
Algorithm~\algo{acp} also uses a probability threshold $q$ which is
lowered at each iteration of its main while loop, and in
each iteration it needs to
estimate reliably probabilities that are at least $q^3$. 
Again, we use the implementation of \algo{min-partial} described
before and in each iteration of the while loop we
set the number of samples as
\begin{equation}
\label{eq:samples-k-median}
r = 
\left\lceil
\frac{12}{q^3\epsilon^2}\ln\left(2 n^3\left(
      1 + \left\lfloor \log_{1+\gamma} \frac{H(n)}{p_L} \right\rfloor
    \right)
  \right)
\right\rceil.
\end{equation}
We have the following result, whose proof, analogous to that of Theorem~\ref{thm:seq1}, is omitted for brevity.
\begin{theorem} \label{thm:impl-kmedian}
The implementation of \algo{acp} terminates after at most $\lfloor\log_{1+\gamma}(H(n)/p_{\rm opt-avg}(k)) \rfloor + 1$ iterations of the while loop and returns a clustering $\tilde{\mathcal{C}}$ with
  \[
    \text{avg-prob}(\tilde{\mathcal{C}}) \ge
    (1-\epsilon) \left(
      \frac{p_{\rm opt-avg}(k)}{(1+\gamma)H(n)}
    \right)^3
  \]
  if $(p_{\rm opt-avg}(k)/H(n))^3 \ge p_L$, with high probability.
\end{theorem}
We observe that, for constant $\gamma$, the overall number of
samples required by our implementation is $\BO{(1/(p^3_{\rm
    opt-avg}(k)\epsilon^2) (\log n+\log\log ((\log n )/p_L))}$. 
Considering that we can safely set $p_L = k/n$, as mentioned at the beginning if the section,
we conclude that 
 $\BO{(1/(p^3_{\rm
    opt-avg}(k)\epsilon^2) \log n}$
samples suffice.


\section{Experiments} \label{sec:experiments}
  
\newcommand{\mcpc}{\algo{mcp}\xspace}
\newcommand{\acpc}{\algo{acp}\xspace}
\newcommand{\paralg}{\algo{Par}\xspace}
\newcommand{\gmm}{\algo{gmm}\xspace}
\newcommand{\mcl}{\algo{mcl}\xspace}

\begin{figure*}[t]
  \centering
  \includegraphics[width=\textwidth]{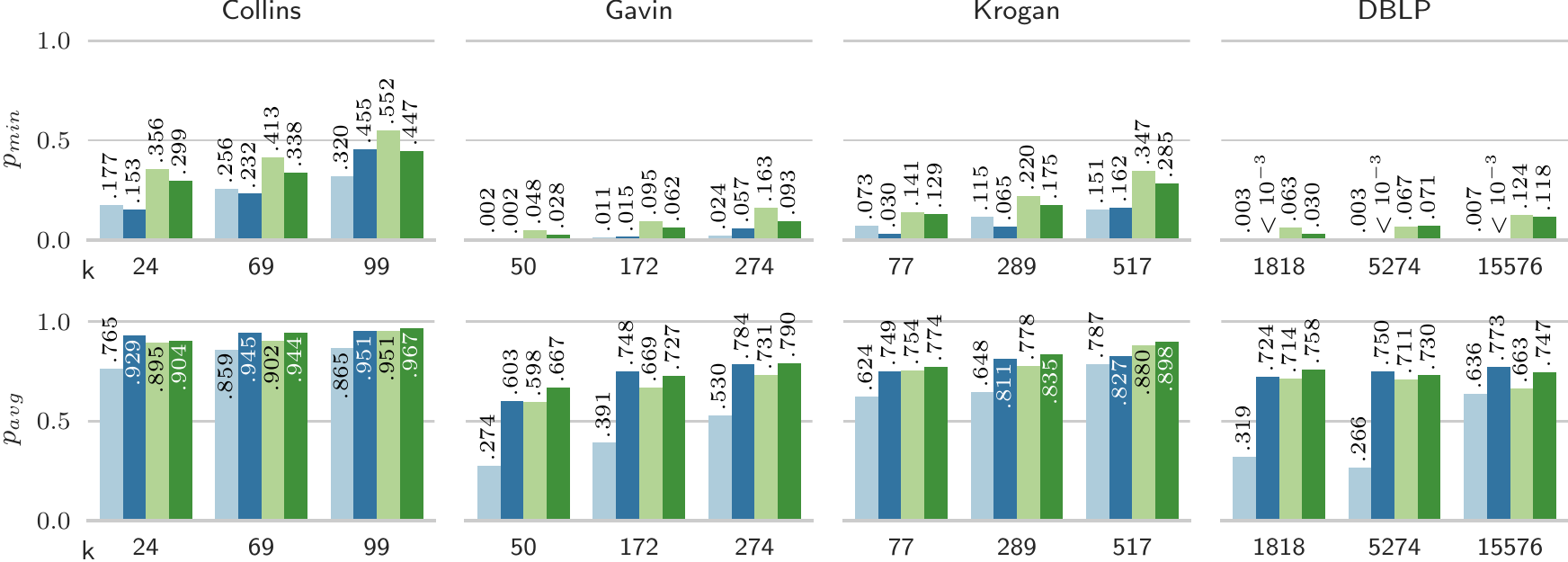}
  \begin{center}
    \begin{minipage}{.65\textwidth}
      \caption{Minimum connection probability ($p_{min}$, top row of plots) and average connection probability ($p_{avg}$, bottom row of plots).}
      \label{fig:p_min_p_avg}
    \end{minipage}
    \begin{minipage}{.3\textwidth}
      \centering
      \includegraphics{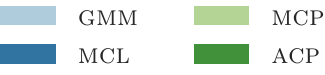}
    \end{minipage}
  \end{center}
\end{figure*}

We experiment with our clustering algorithms \mcpc and \acpc along two different lines.
First, in Subsection~\ref{sec:exp1}, we compare the quality
of the obtained clusterings against those  returned by well-established clustering
approaches in the literature on four uncertain graphs derived by three PPI networks
and a collaboration network. Then, in Subsection~\ref{sec:exp2}, we provide an
example of applicability of uncertain clustering as a predictive tool to spot so-called
\emph{protein complexes} in one of the aforementioned PPI networks.

\begin{table}[t]
  \caption{Graphs considered in our experiments. The number of nodes and the number of edges in the largest connected component are shown.}
  \label{tab:datasets}
  \centering
  \begin{tabular}{l rr}
    \toprule
    graph & nodes & edges \\
    \midrule
    \texttt{Collins} & 1004 & 8323 \\
    \texttt{Gavin} & 1727 & 7534 \\
    \texttt{Krogan} & 2559 & 7031 \\
    \texttt{DBLP} & 636751 & 2366461 \\
    \bottomrule
  \end{tabular}
\end{table}

The characteristics of the four different graphs used in our experiments
are summarized in Table~\ref{tab:datasets}.
Three graphs are PPI networks, with different distributions of edge
probabilities: \texttt{Collins}~\cite{CollinsEtAl07}, mostly comprising
high-probability edges; \texttt{Gavin}~\cite{GavinEtAl06}, where most
edges are associated to low probabilities, and the CORE network
introduced in~\cite{KroganEtAl06} (\texttt{Krogan} in the following),
which has one fourth of the edges with probability greater than 0.9,
and the others almost uniformly distributed between 0.27 and 0.9.
To exercise a larger spectrum of cluster granularities,
we target clusterings only for the largest connected component
of each graph. As a computationally more challenging instance, we also experiment with
a large connected subgraph of the DBLP collaboration network  with edge probabilities
obtained with the same procedure of~\cite{PotamiasBGK10}: each node
is an author, and two authors are connected by an edge if they are
co-authors of at least one journal publication. The probability of such
an edge is $1 - \exp\{-x/2\}$, where $x$ is the number of
co-authored journal papers. Consequently, a single collaboration
corresponds to an edge with probability $0.39$, and 2 and 5
collaborations correspond to edges with probability $0.63$ and $0.91$,
respectively.
Roughly $80\%$ of the edges have probability $0.39$, $12\%$ have
probability $0.63$ and the remaining $8\%$ have a higher probability.
While finding an accurate probabilistic model of the interactions
between authors is beyond the scope of this paper, the intuition behind
the choice of this distribution is that authors that are likely to
collaborate again in the future share an edge with large probability.

We implemented our algorithms in C++, with
the Monte Carlo sampling of possible worlds performed in parallel using OpenMP.
The code and data, along with instructions to reproduce the results presented in this section, are publicly available\footnote{
  \url{https://github.com/Cecca/ugraph}}.
When running both \mcpc and \acpc, we set $\gamma = 0.1$.  To optimize
the execution time, we set the probability threshold $q$ of
Algorithms~\ref{alg:k-center} and~\ref{alg:k-median} to $q_i = \max\{1
- \gamma\cdot 2^i, p_L\}$ in iteration $i$, with $p_L= 10^{-4}$.  Once
$q_i$ equals $p_L$ or is such that the associated clustering covers
all nodes, we perform a binary search between $q_i$ and $q_{i-1}$ to
find the final probability guess, stopping when the ratio between the
lower and upper bound is greater than $1-\gamma$.  
This procedure is essentially equivalent, up to constant factors in the final value of $q$,
to decreasing $q$ geometrically as is done in
Algorithms~\ref{alg:k-center} and~\ref{alg:k-median}, thus the
guarantees of Theorems~\ref{thm:seq1} and~\ref{thm:impl-kmedian} hold.
In the implementation of \acpc, we decided to invoke
\algo{min-partial} with parameters $(\mathcal{G},k,q,1,q)$ 
rather than $(\mathcal{G},k,q^3,n,q)$.  While this
setting does not guarantee the theoretical bounds stated in
Theorem~\ref{thm:impl-kmedian}, the values were chosen after testing
different combinations of the two parameters and finding that the
chosen setting provides better time performance while still returning
good quality clusterings in all tested scenarios.  In particular, we
found out that higher values of parameter $\alpha$ yielded similar scores,
albeit with a lower variance.  For both implementations, we verified
that setting parameter $\gamma$, which essentially controls a time/quality
tradeoff, to values smaller than $0.1$ increases the running time
without increasing the quality of the returned
clusterings  significantly.
%
Considering that the sample sizes defined in
Section~\ref{sec:implementation} are derived to tolerate very
conservative union bounds, we verified that in practice starting the
progressive sampling schedule from 50 samples always yields very
accurate probability estimates.
We omit the results of these preliminary experiments for lack of space.
Both our code and \algo{mcl} were
compiled with GCC 5.4.0, and run on a Linux 4.4.0 machine equipped
with an Intel I7 4-core processor and 18GB of RAM. Each figure we
report was obtained as the average over at least 100 runs, with the
exception of the bigger \texttt{DBLP} dataset, where only  5 runs were
executed for practicality.

\subsection{Comparison with other algorithms}
\label{sec:exp1}

A few remarks on the algorithms chosen for (or excluded from) the
comparison with \algo{mcp} and \algo{acp} presented in this subsection
are in order.  We do not compare with the clustering algorithm for
uncertain graphs devised in~\cite{KolliosPT13} because it does not
allow to control the number $k$ of returned clusters, which is central
in our setting.  (A comparison between \mcpc and the algorithm by
\cite{KolliosPT13} is offered in the next subsection with respect to a
specific predictive task.)  Also, we were forced to exclude from the
comparison the clustering algorithms of~\cite{LiuJAS12}, explicitly
devised for uncertain graphs, since the source code was not made
available by the authors and the  algorithms
do not lend themselves to a straightforward implmenetation.
Among the many clustering algorithms for
deterministic weighted graphs available in the literature, we selected
\mcl~\cite{Dongen08} since it has been previously employed in the
realm of uncertain graphs by using edge probabilities as weights.
Finally, we compare with an adaptation of the $k$-center approximation
strategy of \cite{Gonzalez85}, dubbed \gmm{}, where the clustering of
the uncertain graph is computed in $k$ iterations by repeatedly
picking the farthest node from the set of previously selected centers,
using the shortest-path distances generated by setting the weight of
any edge $e$ to $w(e) = \ln(1/p(e))$.  Considering the modest
performance of \gmm{} observed in our experiments, and the fact that
it has been repeatedly stated in previous works
\cite{LiuJAS12,PotamiasBGK10} that the application of shortest-path
based deterministic strategies to the probabilistic context yields
unsatisfactory results, we did not deem necessary to extend the
comparison to other such strategies.

\begin{figure*}[t]
  \centering
  \includegraphics[width=\textwidth]{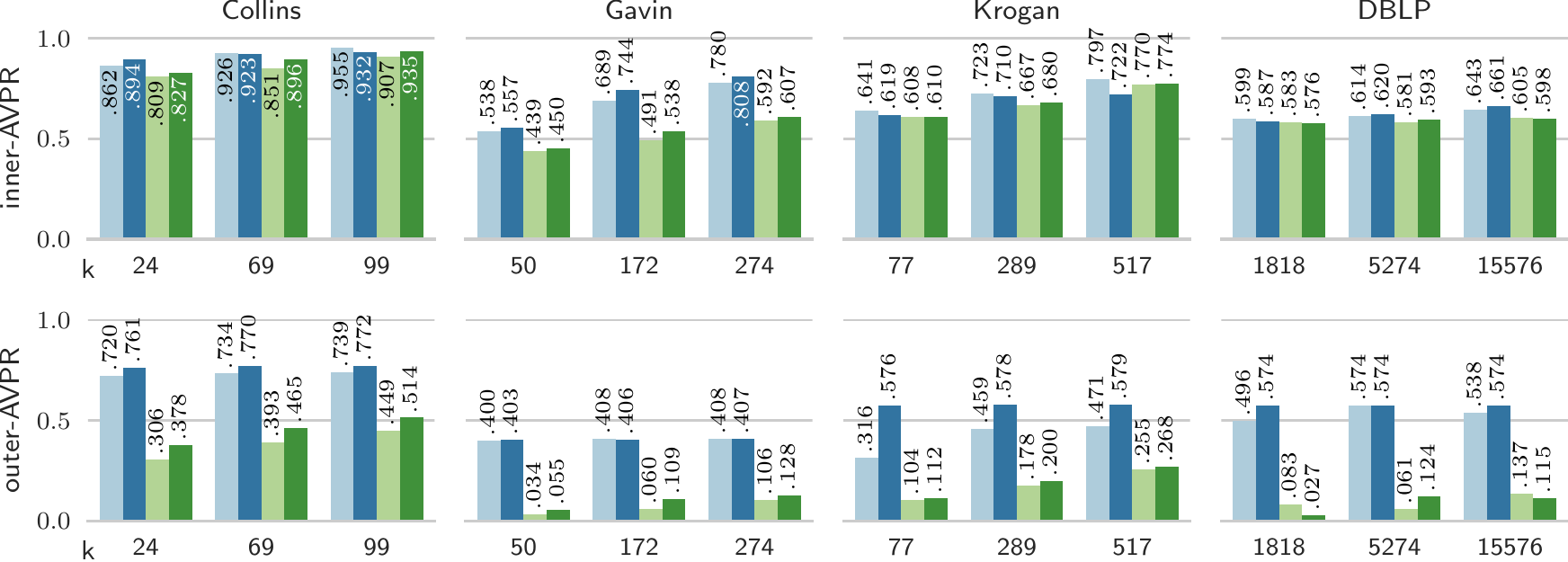}
  \begin{center}
    \begin{minipage}{.65\textwidth}
      \caption{Inner and outer Average Vertex Pairwise Reliability. For the inner-AVPR metric higher is better, for the outer-AVPR metric lower is better.}
      \label{fig:avpr}
    \end{minipage}
    \begin{minipage}{.3\textwidth}
      \centering
      \includegraphics{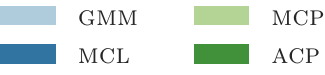}
    \end{minipage}
  \end{center}
\end{figure*}

We base our comparison both on our defined cluster quality metrics and
on other ones. Namely, for each clustering returned by the various
algorithms, we compute the minimum connection probability of any node
to its center\footnote{\small For \algo{mcl}, when computing this
  metric we consider as cluster centers the \emph{attractor nodes} as
  defined in \cite{Dongen08}.} (denoted simply as $p_{min}$), and the
average connection probability of all nodes to their respective
centers (denoted as $p_{avg}$).
Furthermore, we consider the inner Average Vertex Pairwise Reliability (also defined in~\cite{LiuJAS12}) which is the average connection probability of all pairs of nodes that are in the same cluster, namely,
\[
  \text{inner-AVPR} = \frac{\sum_{i=1}^\tau\sum_{u,v \in C_i} \Pr(u \sim v)}{\sum_{i=1}^\tau \binom{|C_i|}{2}},
\]
and the outer Average Vertex Pairwise Reliability, which is the average connection probability of pairs of nodes in different clusters, namely, 
\[
  \text{outer-AVPR} = \frac{\sum_{i=1}^\tau\sum_{u \in C_i, v\notin C_i} \Pr(u \sim v)}{
    \sum_{i=1}^\tau |C_i| \cdot |V \setminus C_i|}.
\]
Intuitively, a good clustering in terms of connection probabilities exhibits a low outer-AVPR and a significantly higher inner-AVPR, indicating that each cluster completely encapsulates a region of high reliability.
Inner-AVPR and outer-AVPR are akin to the internal and external \emph{cluster density} measures used in the setting of deterministic graphs~\cite{Schaeffer07}.
We also compare the algorithms in terms of their running time.

Recall from Section~\ref{sec:intro} that the number of clusters
computed by \algo{mcl} cannot be controlled accurately, but it is
influenced by the inflation parameter.  Therefore, for each graph in
Table~\ref{tab:datasets}, we run \algo{mcl} with inflation set to 1.2, 1.5,
and 2.0 for protein networks, and 1.15, 1.2, and 1.3 for
\texttt{DBLP}, so to obtain a reasonable number of clusters.  We then
run the other algorithms with a target number $k$ of clusters matching
the granularity of the clustering returned by each \algo{mcl} run.
Note that, in terms of running time, this setup favors \mcl{}: if we
were instead given a target number of clusters, we would need to
perform a binary search over the possible inflation values to
make \mcl{} match the target, and the running time for \mcl{} would
become the sum of the times over all these search trials.

As expected, with respect to the $p_{min}$ metric
(Figure~\ref{fig:p_min_p_avg}, top) \mcpc is always better than all other
algorithms.  In particular, on the \texttt{DBLP} graph, both \gmm{}
and \mcl find clusterings with $p_{min}$ very close to zero $(<
10^{-3})$, meaning that there is at least one pair of nodes in the
same cluster with almost zero connection probability. In contrast,
\mcpc finds clusterings with $p_{min}$ very close to $0.1$ or
larger. The inferior performance of \gmm{}, a clustering algorithm
aiming at optimizing an extremal statistic like $p_{min}$ in a
deterministic graph, provides evidence that naive adaptations of
deterministic clustering algorithms to the probabilistic scenario
struggle to find good solutions. Further evidence in this direction is
provided by the fact that \acpc{} yields higher values of $p_{min}$
than both \gmm{} and \mcl.

For what concerns the $p_{avg}$ metric (Figure~\ref{fig:p_min_p_avg}, bottom),
observe that, being an average, the metric hides the presence of low
probability connections in the same cluster, which explains the much
higher values returned by all algorithms compared to $p_{min}$.
Somewhat surprisingly, we have that \mcl and \acpc{} have comparable
performance. Recall, however, that \mcl does not guarantee total control on the 
number of clusters, which makes \acpc{} a more
flexible tool. In all cases, the \gmm algorithm finds clusters with a
$p_{avg}$ that is lower than the one obtained by the other algorithms.
We observe that the experiments provide evidence that the actual
values of $p_{avg}$ obtained with \acpc are arguably much higher than
the theoretical bounds proved in Theorem~\ref{thm:impl-kmedian}, and
we conjecture that this also holds for $p_{min}$.
Consider now the inner- and outer-AVPR metrics (Figure~\ref{fig:avpr},
resp. top and bottom). For all graphs, the clusterings computed by \mcpc
and \acpc feature an inner-AVPR comparable to \gmm and \mcl, but a
considerably lower outer-AVPR, which is a desirable property of a clustering, as
mentioned earlier. Conversely, for a given graph and value of $k$,
\mcl and \gmm compute clusterings where the inner- and outer-AVPR scores
are similar. This fact suggests that the other clustering
strategies are driven more by the
topology of the graphs rather than by the connection probabilities.

\begin{figure*}
  \centering
  \includegraphics[width=\textwidth]{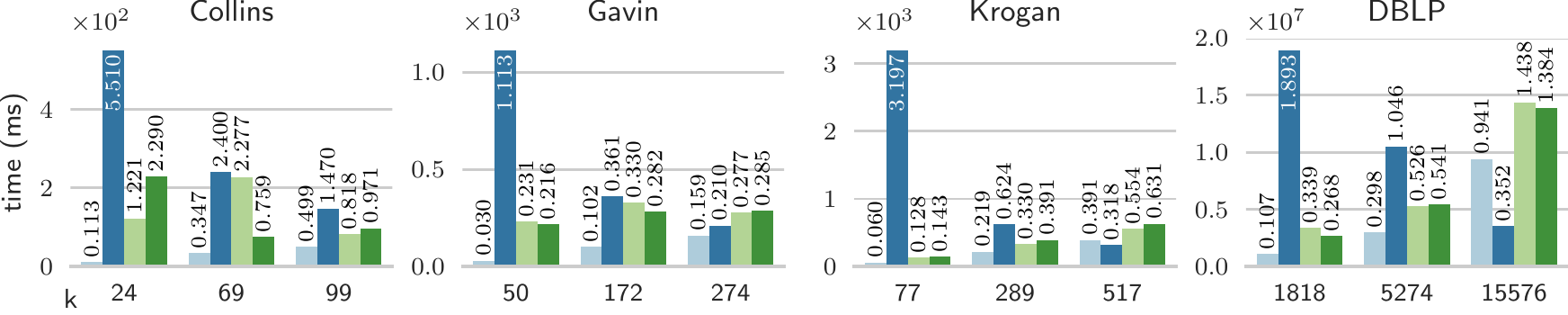}
  \begin{center}
    \begin{minipage}{.4\textwidth}
      \caption{Running times in milliseconds.}
      \label{fig:time}
    \end{minipage}
    \begin{minipage}{.25\textwidth}
      \centering
      \includegraphics{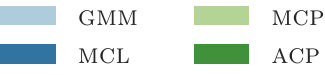}
    \end{minipage}
  \end{center}
\end{figure*}

\iflong
As for the running times (Figure~\ref{fig:time}), \gmm is almost
always the fastest algorithm, due to its obliviousness to the possible
world semantics, which entails resorting to expensive Monte Carlo sampling, and its running
time grows linearly in $k$.  On the other hand, the running time of \mcl exhibits an
inverse dependence on $k$ since clusterings for low values of $k$,
which are arguably more interesting in practical scenarios, are the
most expensive to compute. Thanks to the use of progressive sampling,
our algorithms feature a running time which is significantly better or
comparable to \mcl, depending on the granularity of the clustering.
With respect to \gmm, our algorithms are slower, but they provide far
better clusterings, as discussed above.

\begin{figure}[t]
  \centering
  \includegraphics[width=\columnwidth]{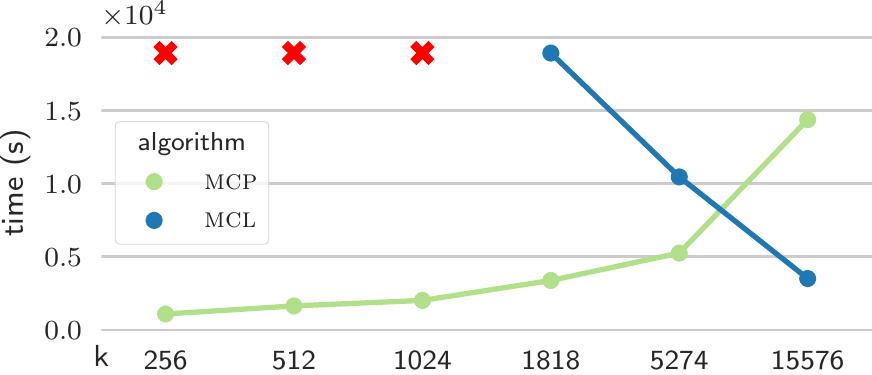}
  \caption{Running time (in seconds) versus $k$ on the DBLP graph.
    Red crosses denote
    runs in which \algo{\sc mcl} failed due to memory issues: we consider the running time for $k=1818$ as a lower bound for the execution time of these runs.}
  \label{fig:scalability}
\end{figure}

Figure~\ref{fig:scalability} shows a more detailed comparison of the
running time of \mcpc and \mcl on the \texttt{DBLP} dataset as a
function of the number $k$ of clusters.
For small values of $k$, \mcl
crashed after exhausting the memory available on the system, whereas
our algorithm is very efficient for such values.
A comparison between  \acpc{}  and \mcl would yield the same
conclusions since the running time of \acpc{} is comparable to that
of \mcpc (see Figure~\ref{fig:time}).

\fi

\subsection{Clustering as a predictive tool}
\label{sec:exp2}
In PPI networks, proteins can be grouped in so-called
\emph{complexes}, that is, groups of proteins that stably interact
with one another to perform various functions in a cell.  Given a PPI
network, a crucial problem is to predict protein pairs belonging to
the same complex.  Our specific benchmark is the \texttt{Krogan}
graph, for which the authors published a clustering with 547 clusters
obtained using \algo{mcl} with a configuration of parameters that
maximizes biological significance~\cite[Suppl. Table
  10]{KroganEtAl06}.  We consider a ground truth derived from the
publicly available, hand-curated MIPS database of protein
complexes~\cite{MewesEtAl04,MIPS} as used in~\cite{KroganEtAl06}.
For the purpose of the evaluation,
we restrict ourselves to proteins appearing in both \texttt{Krogan}
and MIPS, thus obtaining a ground truth with 3,874 protein pairs. The
input to the clustering algorithms is the entire \texttt{Krogan}
graph.  We evaluated the returned clusterings in terms of the
confusion matrix. Namely, a pair of proteins assigned to the same
cluster is considered a true positive if both proteins appear in the
same MIPS complex, and a false positive otherwise.

For brevity, we restrict ourselves to exercise \mcpc and \acpc considering
only $d$-connection probabilities (see Section~\ref{sec-depthlimit})
for different values of $d$, and by setting $k=547$ so as to match the
cardinality of the reference clustering from \cite{KroganEtAl06}.  The
idea behind the use of limited path length is that we expect proteins
of the same complex to be connected with a high probability and, at
the same time, to be topologically close in the graph.  We do not
report results for $d=1$ since there is no clustering of the
\texttt{Krogan} graph with $d=1$ and $k=547$. We compare the True
Positive Rate (TPR) and the False Positive Rate (FPR) obtained by
\mcpc with different values of $d$ against those obtained with the
\mcl-based clustering of \cite{KroganEtAl06}, and with the clustering
computed by the algorithm in~\cite{KolliosPT13} (dubbed \algo{kpt} in
what follows).  The results are shown in
Table~\ref{tab:ground-truth}. We observe that, for small values of
$d$, our algorithm is able to find a clustering with scores similar to
the clustering of \cite{KroganEtAl06}, while higher values of $d$
yield fewer false negatives at the expense of an increased number of
false positives.
Note that \acpc is slightly better than \mcpc w.r.t. the TPR, and \mcpc
tends to be more conservative when it comes to the FPR.
Furthermore, the FPR performance of \acpc degrades more quickly as the depth
increases. This is because \acpc optimizes a measure that allows clusters
where a few nodes are connected with low probability to their centers,
and this effect amplifies as the depth increases.
Thus, we can choose to use \mcpc if we want to keep the number of false
positives low, or \acpc to achieve a higher TPR.
Observe that a moderate number of false positives
may be tolerable, since the corresponding protein pairs can be the
target of further investigation to verify unknown protein
interactions.  Also,  \mcpc, \acpc, and \mcl yield TPRs substantially
higher than \algo{ktp}\footnote{\small We remark that the performance
  of \algo{ktp} reported here differs from the one reported in
  \cite{KolliosPT13} since the ground truth considered in that paper
  only comprises pairs of proteins that appear in the same complex in
  the MIPS database and are connected by an edge in the
  \texttt{Krogan} graph, which clearly amplifies the TPR.
}.
This experiment supports our intuition that considering topologically close
proteins, while aiming at high connection probabilities, makes our
algorithm competitive with state-of-the-art solutions in the predictive
setting.

\begin{table}[t]
  \caption{\algo{\sc mcp} and \algo{\sc acp} with limited path length against
    \algo{\sc mcl} and \algo{\sc kpt} on \texttt{Krogan} w.r.t. the MIPS ground truth.}
  \label{tab:ground-truth}
  \centering
  \begin{tabular}{lc rr rr}
    \toprule
    &
    & \multicolumn{2}{c}{TPR}
    & \multicolumn{2}{c}{FPR} \\
    \midrule
    & & \algo{mcp} & \algo{acp} & \algo{mcp} & \algo{acp} \\

    \multirow{5}{*}{Depth $\displaystyle\left\{\begin{tabular}{c} \\ \\ \\ \\ \end{tabular}\right.$ \hspace{-3em}}
    & 2    &  0.344  &  0.384 &  0.003 &  0.006  \\
    & 3    &  0.416  &  0.459 &  0.012 &  0.078  \\
    & 4    &  0.429  &  0.585 &  0.147 &  0.419  \\
    & 6    &  0.695  &  0.697 &  0.604 &  0.633  \\
    & 8    &  0.737  &  0.730 &  0.678 &  0.647  \\

    \midrule
    \multicolumn{2}{c}{\algo{mcl}~\cite{KroganEtAl06}}
    & \multicolumn{2}{c}{0.423}
    & \multicolumn{2}{c}{0.002} \\
    \multicolumn{2}{c}{\algo{kpt}~\cite{KolliosPT13}}
    & \multicolumn{2}{c}{0.187}
    & \multicolumn{2}{c}{$6.3 \cdot 10^{-4}$} \\
    \bottomrule
  \end{tabular}
\end{table}


\section{Conclusions}
\label{sec:conclusions}
We presented a number of algorithms for clustering uncertain graphs
which aim at maximizing either 
the minimum or the average connection probability of
a node to its cluster's center. Unlike previous approaches, our
algorithms feature provable guarantees on the clustering quality, 
and afford efficient implementations and an effective
control on the number of clusters.  We also provide an open
source implementation of our algorithms which compares favorably with
algorithms commonly used when dealing with uncertain graphs.

Our algorithms \mcpc and \acpc are based on objective
functions capturing different statistics (minimum, average) of
connection probabilities, that are not directly comparable.
As such, there cannot be clear guidelines for choosing one
of the two in practical  scenarios. Considering that \mcpc and \acpc
feature comparable running times, the most reasonable approach could
be to apply both and then choose the most ``suitable'' returned clustering,
where suitability could be measured, for instance, through the use of
external metrics such as inner/outer-AVPR.

Several challenges remain open for future research.  
From a complexity
perspective, the conjectured NP-hardness of the ACP problem and,
possibily, inapproximability results for both the MCP and ACP problems
remain to be proved.  More interestingly, there is still ample room
for the development of practical algorithms featuring better
analytical bounds on the approximation quality and/or faster
performance, as well as for the investigation of other clustering
problems on uncertain graphs.


\balance
\bibliographystyle{abbrv}
\bibliography{references}

\end{document}